\documentclass[12pt]{amsart}
\usepackage{amsmath}
\usepackage{amssymb}
\usepackage{amsthm}
\usepackage{mathrsfs}
\usepackage{comment}
\usepackage{hyperref}

\usepackage[all,cmtip]{xy}\usepackage{xcolor}
\usepackage{enumerate}
\usepackage{bm}
\usepackage{dsfont}
\usepackage{mathtools}
\usepackage[cal=euler]{mathalfa}
\usepackage[top=1in, bottom=1.25in, left=1.25in, right=1.25in]{geometry}
\usepackage{parskip}

\hypersetup{colorlinks=true,linkcolor=magenta,citecolor=blue}

\DeclareMathAlphabet{\mathpzc}{OT1}{pzc}{m}{it}

\theoremstyle{plain}
\newtheorem{theorem}{Theorem}[section]
\newtheorem*{theorem*}{Theorem}

\newtheorem{lemma}[theorem]{Lemma}
\newtheorem*{claim*}{Claim}
\newtheorem{proposition}[theorem]{Proposition}

\newtheorem{corollary}[theorem]{Corollary}

\theoremstyle{definition}

\newtheorem{example}[theorem]{Example}

\newtheorem{remark}[theorem]{Remark}

\numberwithin{equation}{section}
\numberwithin{figure}{section}

\newcommand{\Gal}{\mathrm{Gal}}

\newcommand{\ignore}[1]{}

\begin{document}

\title[Division algebras and linear maximum rank distance  codes]
{A skew polynomial framework for constructing division algebras and linear maximum rank distance  codes}

\author{S. Pumpl\"un}

\email{susanne.pumpluen@nottingham.ac.uk}

\address{School of Mathematical Sciences\\
University of Nottingham\\ University Park\\ Nottingham NG7 2RD\\
United Kingdom }



\maketitle
\date{\today}

\begin{abstract}
We construct division algebras and linear maximum rank distance (MRD) matrix codes using skew polynomials over fields.  The non-unital division algebras we obtain  generalize several prominent constructions: Sheekey's twisted cyclic pre-semifields, i.e. the pre-semifields associated  with Jha-Johnson semifields and the semifields associated with Albert's generalized twisted fields.
   Our linear MRD codes generalize the constructions of Lobillo, Santonastaso and Sheekey. We  present criteria for these algebras to be division algebras, respectively, for when these codes have maximum rank, and compare isotopic division algebras that appear throughout recent and classical literature.  We compute some of their invariants.
\end{abstract}

\section{Introduction}

Skew polynomials in  $R= K[t;\sigma]$ or even more general skew polynomial rings $R= K[t;\sigma,\delta]$ have recently been  used in diverse constructions of both nonassociative division algebras \cite{LS, NewS, PT2} and  linear  maximum rank distance (MRD) matrix codes \cite{NewS, Sheekey16,  Sheekey, PT2, DT2020, LTZ}.  Both the construction of a class of non-unital division algebras and the construction of MRD codes that we present in this paper present a unifying frame for previous work in
\cite{NewS, PT2, Sheekey} which was initiated by Petit's construction of unital algebras in 1966 \cite{P66}.

Let $K/F$ be a cyclic Galois field extension of degree $n$ with Galois group generated by $\sigma$ and let $R=K[t;\sigma]$. Define $R_m=\{ g\in R\,|\, \deg(g)<m \}$. For our algebra construction, we will define a multiplication on $R_m$ which makes $R_m$ into a non-unital algebra over $F_0$.

Our approach is the following. We choose a monic irreducible skew polynomial $f\in R$  of degree $m$, some $\nu\in K$, and fix an integer $i_0\in \{0,\dots, m-1\}$.

  Let  $F'$ be a field,  put $F_0=F\cap F'$ and assume that $F/F_0$  is a  field extension of finite degree. Choose an $F'$-linear bijective map $\rho:K\to K$. The definition of the multiplication on $R_m$ depends on the choices of $f$, $\rho$, $\mu\in K$ and $i_0$.

Let $b(t), c(t)\in R_m$, and write $b(t)=\sum_{i=1}^{m-1} b_it^i $. We  define
$$b(t)\circ c(t)= (b(t)+\nu\rho(b_{i_0})t^m )c(t)\,\,{\rm mod}_r f,$$
where the right hand side is the remainder of what we get when we divide the product $(b(t)+\nu\rho(b_{i_0})t^m )c(t)$ by $f$ on the right.
This multiplication makes $R_m$ into a non-unital nonassociative ring $(R_m,\circ)$
which we denote by  $\mathbb{S}_f(K,i_0,\nu,\rho)$. The ring $\mathbb{S}_f(K,i_0,\nu,\rho)$ can also be viewed as an algebra of dimension $m[K:F_0]$ over $F_0$.
 Sheekey's twisted cyclic semifields and the algebras  subsequently treated in  \cite{NewS} appear as the special case that $i_0=0$ and $\rho\in {\rm Aut}(K)$
 for suitable choices of $\nu\in K^\times$.

Thus we extend the existing approaches in two directions by asking the following two questions: When $i_0=0$ as in \cite{Sheekey, PT2, NewS},
  what happens if we  employ any bijective $F'$-linear map $\rho:K\to K$ as in our definition, and not just ring automorphisms of $K$?  For instance,  every bijective $F$-linear map $\rho:K\to K$ can be written as a linear combination  $\rho(y)= \sum_{i=0}^{n-1}y_i\sigma^i(y)$  of Galois automorphisms $\sigma^i$
  for suitable $y_i\in K$. Analogously, when $K/F'$ is a cyclic Galois field extension then every bijective $F'$-linear map $\rho:K\to K$  has a similar form; the  choice of maps we can use in our construction is hence substantially larger.

  Secondly, what happens if we allow $i_0$ to be any element in $ \{ 0,\dots, m-1 \}$ and not limit ourselves to $i_0=0$ as was done in all previous constructions?

We first  establish an abstract criterium for $\mathbb{S}_f(K,i_0,\nu,\rho)$ to be a division algebra.

\begin{theorem*} (Theorem \ref{thm:division})
Let $f\in R$ be a monic and irreducible polynomial of degree $m$.
The algebra $\mathbb{S}_f(K,i_0,\nu,\rho)$  is a division algebra over $F_0$ if and only if the set
 $$P_{i_0}=\lbrace d_0+d_1t+\dots+d_{m-1}t^{m-1}+\nu\rho(d_{i_0})t^{m}\,|\, d_i\in K \rbrace,$$ does not contain any polynomial that is similar to $f$.
\end{theorem*}

This translates into a criterium that  requires ``only'' the knowledge of the monic polynomials that are similar to $f$, and of the set
$$C_{f,{i_0}}=\{ f_{i_0}\,|\, f_{i_0}\not=0, s (t)=f_0+f_1t+\dots+f_{m-1}t^{m-1}+t^m \text{ is monic and similar to } f\}$$
of all nonzero $i_{0}$th coefficients of the monic polynomials that are similar to $f(t)$ (when $K$ is infinite, this may obviously be an infinite set).

\begin{theorem*} (Theorem \ref{thm:division2})
Let $f\in R$ be a monic irreducible polynomial of degree $m$. Then
$\mathbb{S}_f(K,i_0,\nu,\rho)$  is a division algebra over $F$ if and only if  for all
 monic polynomials  $g(t)=g_0+g_1t+\dots+ g_{i_0}t^{i_0}+\dots +t^m\in R$ that are similar to $f$, either $\nu\rho(g_{i_0})\not =1$ and $g_{i_0}\not=0$, or $g_{i_0}=0$.
\end{theorem*}

 The criteria we obtain for our algebras to be division algebras, and in particular pre-semifields when $K$ is finite, are all derived from this observation.
 These criteria are special cases of Proposition \ref{thm:normI} and Theorem \ref{thm:normstrong} and can be summed up as follows.

 \begin{theorem*}
 Let $K$ be a finite field,  $f\in R$  a monic irreducible polynomial of degree $m$, and define
 $$N_{f,{i_0}}=N_{K/F}(C_{f,{i_0}}).$$
 Suppose one of the following holds for all $g(t)=g_0+g_1t+\dots+g_{m-1}t^{m-1}+\nu\rho(g_{i_0})t^m\in P_{i_0}$ with $g_{i_0}\not=0$, then $\mathbb{S}_f(K,i_0,\nu,\rho)$ is a pre-semifield.
\\ (i)
$ g_{i_0}\nu^{-1}\rho(g_{i_0})^{-1}\not \in C_{f,{i_0}}$,
 \\ (ii) $ N_{K/F}(g_{i_0}\nu^{-1}\rho(g_{i_0})^{-1})\not\in N_{f,{i_0}}$.
\end{theorem*}

 We  give an alternative condition to  \cite[Theorem 7]{Sheekey} for Sheekey's twisted cyclic algebras to be pre-semifields and also recover \cite[Theorem 7]{Sheekey} as special cases in
 Corollaries \ref{cor:Sh} and \ref{cor:Sh2}.
 Along the way,
 we generalize Theorem \cite[Theorem 5.4]{NewS} (cf. Theorem \ref{thm:algebra2}), and show that the algebras studied in Theorem \cite[Theorem 5.4]{NewS} are isotopic to Petit algebras, therefore only new when algebras are studied up to isomorphisms. Furthermore,
 we generalize \cite[Theorem 5.7]{NewS} and also show that the algebra constructed in \cite[Theorem 5.7]{NewS} is not unital (Theorem \ref{thm:algebra3}), which was claimed, and instead only have a left unit element.

Division algebras over $F_0$  trivially yield $F_0$-linear maximum rank distance (MRD) matrix  codes. Hence
the matrices representing the right  multiplication in a division algebra $\mathbb{S}_f(K,i_0,\nu,\rho)$  yield an $F_0$-linear MRD code, where all matrices have full column rank (and analogously, so do the matrices representing the left  multiplication).

More general linear MRD-codes are constructed in the second part of the paper,  generalizing the constructions  in \cite{NewS, PT2, Sheekey}.
Our approach is the following. We choose again a monic irreducible skew polynomial $f\in R$  of degree $m$, some $\nu\in K^\times$, and fix an integer $i_0\in \{0,\dots, m-1\}$. We assume that $f$ is not two-sided, that means does not generate a two-sided ideal in $R$.
 Let $h\in R$ be the minimal central left multiple
 of $f$. Let $k$ be the number of irreducible factors of $h\in R$. Let $B={\rm Nuc}_r(\mathbb{S}_f)$ be the right nucleus of the nonassociative unital Petit algebra $\mathbb{S}_f=R/Rf$ \cite{P66}. Let $l$ be an integer such that  $1<l<k$.
Then the image $\mathbb{S}_{n,m,l}(K,i_0,\nu,\rho, f)$ of the set
 $$P_{i_0,l}=\lbrace d_0+d_1t+\dots+d_{lm-1}t^{lm-1}+\nu\rho(d_{i_0})t^{lm}\,|\, d_i\in K \rbrace,$$
in $R/Rh$ defines an $F_0$-linear matrix code $\mathcal{C}_{i_0, n,m,l}$ in  $M_{k}(B)$ that has maximum rank under certain conditions.

 This construction again extends the existing approaches in two directions: we employ any  $F'$-linear map $\rho:K\to K$ instead of requiring that  $\rho\in {\rm Aut}(K)$ is a ring automorphism, and choose any $i_0\in\{0,\dots,lm-1\}$ instead of setting $i_0=0$.
 Sheekey's original MRD-code construction is the special case that $i_0=0$, $F$ is finite, and $\rho=\sigma^i\in {\rm Gal}(K/F)$.

We  establish an abstract criterium for the linear codes $\mathcal{C}_{i_0, n,m,l}$ to be   MRD matrix codes
with minimum distance $k-l+1$.

\begin{theorem*} (Theorem \ref{thm:MRD})
The set
$\mathbb{S}_{n,m,l}(K,i_0,\nu,\rho, f)$  defines an $F_0$-linear MRD-code $\mathcal{C}_{i_0, n,m,l}$  in $M_{k}(B)$
 with minimum distance $k-l+1$, if and only if the set $P_{i_0,l}$ does not contain any polynomial of degree $lm$ with $l$ irreducible factors which are all similar to $f$.
\end{theorem*}

We obtain some necessary conditions involving the norm of $K/F_0$ for when the image of $P_{i_0,l}$ results in an MRD-code.

We make some first attempts to compute the nuclei of the pre-semifields and the left and right idealisers, centraliser and center  of the codes  in Section \ref{sec:nuc}.  In Theorem \ref{theorem 9IIa}  we compute these invariants  for
 the special case that $i_0=0$ and in Theorem \ref{theorem 9II}, we obtain partial results for these invariants for finite base fields and any $i_0$. We finish with some first examples in Section \ref{sec:finite fields}. More explicit examples, in particular over larger finite fields, will need to be computed using use a computer algebra system, as the relevant computations quickly become non-tractable when done manually.

\section{Preliminaries}

\subsection{Nonassociative algebras} \label{subsec:1}

Let $F$ be a field and let $A$ be an $F$-vector space. $A$ is an \emph{algebra} over $F$, if there exists an $F$-bilinear
map $A\times A\mapsto A$, $(x,y) \mapsto x \cdot y$, usually denoted
by juxtaposition, the  \emph{multiplication} of $A$. An algebra
$A$ is  \emph{unital} if there is an element in $A$, denoted by 1, such that $1\cdot x=x\cdot 1=x$ for all $x\in A$.
A finite-dimensional algebra over a finite field is called a \emph{pre-semifield}, if it has a unit element it is called a \emph{semifield}.
 We call a nonassociative algebra a \emph{proper} nonassociative algebra, if it is not associative.

Every nonassociative algebra is a  nonassociative ring $(A,+,\cdot)$ with the distributivity laws satisfied by definition of its multiplication. Conversely, every nonassociative ring can be viewed as a nonassociative algebra over its center, or any subring of it.

For a nonassociative ring or  algebra $A$, the associativity of $A$ is measured by the {\it
associator} $[x, y, z] = (xy) z - x (yz)$. The \emph{left, middle} and \emph{right nucleus} of
$A$ is defined as ${\rm Nuc}_l(A) = \{ x \in A \, \vert \, [x, A, A]
= 0 \}$, ${\rm Nuc}_m(A) = \{ x \in A \,
\vert \, [A, x, A]  = 0 \}$ and  ${\rm
Nuc}_r(A) = \{ x \in A \, \vert \, [A,A, x]  = 0 \}$. All three are associative
subalgebras of $A$. Their intersection
 ${\rm Nuc}(A) = \{ x \in A \, \vert \, [x, A, A] = [A, x, A] = [A,A, x] = 0 \}$ is called the {\it nucleus} of $A$. Whenever one of the elements
 $x, y, z$ lies in
${\rm Nuc}(A)$, we know that $x(yz) = (xy) z$. The  {\it commuter} of $A$ is defined as ${\rm
Comm}(A)=\{x\in A\,|\,xy=yx \text{ for all }y\in A\}$ and the {\it
center} of $A$ is defined as ${\rm C}(A)=\text{Nuc}(A)\cap  {\rm Comm}(A)$ \cite{Sch}.

A nonassociative ring $A\not=0$ has no zero divisors if and only if
$L_a$ and $R_a$ (left and right multiplication with $a$ in $A$) are injective for all $0\not=a\in A$.   An algebra (resp., a nonassociative ring) $A\not=0$  is called a {\it division algebra} (resp., a {\it division ring}) if for all $a\in A$, $a\not=0$, both
the left and right multiplication $L_a$ and $R_a$, are bijective for all $0\not=a\in A$.

 We can apply Kaplanski's trick to a non-unital division algebra $A$ to obtain an isotopic unital division algebra:
we simply define a new multiplication via $x  y=(R_v^{-1}x) (L_u^{-1}y)$ for some $u,v\in A$.
  Different choices of $u, v\in A$  yield isotopic unital division algebras.
  Semifields (i.e., unital division algebras over a finite field) and pre-semifields (i.e., division algebras over a finite field without a unit element) are usually only classified up to isotopy, so the choices of $u,v$ are not relevant in that setting.

  Isotopic unital algebras have isomorphic left, middle and right nuclei, so nuclei are invariant subalgebras under isotopy.
For isotopic pre-semifields $A$ and $A'$, we still know that the dimensions of the nuclei of $A$ and $A'$ are invariant under isotopy, hence are considered as invariants of pre-semifields and semifields.
For division algebras over infinite fields, this is not the case. An easy example is the unital algebra $A=\mathbb{C}$ of complex numbers, which has left, middle and right nucleus $\mathbb{C}$, while the isotopic and non-unital division algebra $A'$ with multiplication given by $x\cdot y=x \bar y$, $\bar y$ the complex conjugate of $y$, has  right nucleus $\mathbb{R}$ and left and middle nucleus  $0$.

\subsection{Skew polynomial rings}\label{subsec:skewpol}
Let $K/F$ be a cyclic Galois field extension of degree $n$ with Galois group generated by $\sigma$. We denote the invertible elements in a field $K$ by $K^\times$.
 The \emph{skew polynomial ring} $R=K[t;\sigma]$
is the ring of polynomials $\{a_0+a_1t+\dots +a_nt^n\,|\, a_i\in K\}$, with term-wise addition and multiplication given by the rule
$ta=\sigma(a)t$ for all $a\in K$  \cite[Chapter I]{J96}. The constant nonzero polynomials $K^\times$ are the units of $R$. Put ${\rm Fix}(\sigma)=\{a\in K\,|\, \sigma(a)=a\}$. The skew polynomial ring $R$ has center $C(R) = F[t^n]\cong F[x]$,
 where $x=t^n$ \cite[Theorem 1.1.22]{J96}.

For $f(t)=a_0+a_1t+\dots +a_nt^n\in R$ with $a_n\not=0$ define ${\rm
deg}(f)=n$ and put ${\rm deg}(0)=-\infty$. Then ${\rm deg}(fg)={\rm deg}
(f)+{\rm deg}(g).$
 An element $f\in R$ is called \emph{irreducible} if  it is not a unit and it has no proper factors, \emph{i.e.} there do not exist
 $g,h\in R$, neither a unit, such that $f=gh$ \cite[p.~2 ff.]{J96}.

 There exists a right division algorithm in $R=K[t;\sigma]$: for all $g,f\in R$, $f\neq 0$, there exist unique $r,q\in R$
  such that ${\rm deg}(r)<{\rm deg}(f)$ and $g=qf+r$ \cite[\S1.1]{J96}.
(Our terminology is the one used by Petit \cite{P66} and
 different from Jacobson's \cite{J96}, who calls what we call right a left division algorithm and vice versa.)
 For all $g \in R$ let $g\, {\rm mod}_r f$ denote the remainder of $g$ upon right division by $f$.

  A skew polynomial  $f\in K[t;\sigma,\delta]$ is  \emph{bounded} if there exists a nonzero skew polynomial
  $f^*  \in K[t;\sigma,\delta]$ such that $K[t;\sigma,\delta]f^* $ is the largest two-sided ideal of $K[t;\sigma,\delta]$ contained in
  $K[t;\sigma,\delta]f$. The polynomial $f^* $ is uniquely determined by $f$ up to scalar multiplication by elements in $K^\times$. $f^*$ is called the \emph{bound} of $f$. In our set-up, all skew polynomials are bounded.

  Two skew polynomials $f$ and $g$ are called \emph{similar} if $R/Rf\cong R/Rg$ as left $R$ -modules, equivalently, there exist $u,v\in R$, such that ${\rm  gcrd}(f,u)=1$, ${\rm gcld}(g,v)=1$ and $gu=vf$. The element $u\in R$ can be chosen such that $u\in R_m$ \cite[p.~15]{J96}.

\subsection{Forms of higher degree}

Let $n\geq 2$ be an integer and let $F$ be a field of characteristic 0 or $>n$.
 Let $N:V\to F$ be a  form of degree $n$ on an $F$-vector space
$V$ of dimension $m$ (i.e., after suitable identification, $N$ is a homogeneous polynomial of degree $n$
in $m$ indeterminates).
A form $N$ of degree $n$ over $F$ is called {\it isotropic},
if there is a non-zero element $x\in V$ such that $N(x)=0$, otherwise it is called {\it anisotropic}.

 An $n$-{\it linear form} over $F$ is a $F$-multilinear map $\theta : V \times
\dots \times V \to F$ on a finite-dimensional vector space $V$ over $F$ which is {\it symmetric}, i.e. $\theta (v_1,
\dots, v_n)$ is invariant under all permutations of its variables.

A {\it form of degree $n$} over $F$ is a map $N:V\to F$  on a finite-dimensional vector space $V$ over $F$
 such that $N(a v)=a^n\varphi(v)$ for all $a\in F$, $v\in V$ and such that the map $N: V \times
\dots \times V \to F$ ($n$-copies) defined by
 $$N(v_1,\dots,v_n)=\frac{1}{n!} \sum_{1\leq i_1< \dots<i_l\leq n}(-1)^{n-l}N(v_{i_1}+ \dots +v_{i_l})$$
(with $1\leq l\leq n$) is a $n$-linear form.
  Any $n$-linear form $\theta : V \times
\dots \times V \to F$ induces a form $N: V\to F$ of degree $n$ via
$N(v)=\theta(v,\dots,v)$.

Two forms $(V_i,N_i)$ of degree $d$, $i=1,2$, are called {\it isometric}
 if there exists a bijective $F$-linear map
$\rho:V_1\to V_2$ such that $N_2(\rho(v))=N_1(v)$ for all $v\in V_1,$ and  {\it similar}  if there exists a bijective $F$-linear map
$\rho:V_1\to V_2$ and some $a\in F^\times$ such that $N_2(\rho(v))=a N_1(v)$ for all $v\in V_1,$

\subsection{Norm isometries and endomorphisms}\label{sec:isometries}
Let $K/F$ be a cyclic Galois field extension of degree $n$ with Galois group  ${\rm Gal}(K/F)=\langle\sigma \rangle$ and with norm $N_{K/F}$ and trace $T_{K/F}$. The norm of a Galois field extension is a  nondegenerate anisotropic multiplicative form of degree $n$.
When $F$ is a field of characteristic 0 or $>n$, we denote  the $n$-linear form associated to $N_{K/F}$ by
$N_{K/F}(x_1,\dots,x_n)$ 
(with $1\leq l\leq n$). Recall that  $N_{K/F}$ is called nondegenerate if
$x = 0$ is the only vector such that $N_{K/F}(x, x_{2}, \dots, x_n) = 0$ for all
 $x_i \in K$, that is $N_{K/F}$ is nondegenerate if and only if $N_{K/F}(x,  w,\dots, w) = 0$ for all
$w\in K$ implies $x=0$.

Take the cyclic algebra  $A_0=(K/F, \sigma, 1)=K[t;\sigma]/(t^n-1)$ of degree $n$ with reduced norm $N_{A_0/F}: A_0 \rightarrow F$,
$N_{A_0/F}(x) = {\rm det}(R_x).$
 Consider $A_0$ as a left $K$-module with basis $\{1,t, \ldots, t^{n-1}\}$.

Let $V$ be an $n$-dimensional $F$-vector space.  View $K\in {\rm Alg}(V)$.
Then every $\rho\in {\rm End}_F(V)$ can be uniquely written in the form
 \begin{equation}\label{iso}
 \rho(x)= \sum_{i=0}^{n-1}y_i\sigma^i(x)
 \end{equation}
 for some $y_i\in K$ and  $\rho$ is bijective  if and only if
 $$N_{A_0}(y_0 + y_1 t + \cdots + y_{n-1}t^{n-1})\not=0,$$
 e.g. cf. \cite[Corollary 2.3]{Pum2025.2}.

  Let $O(N_{K/F})$ denote the group of isometries, and $S(N_{K/F})$ the group of similarities of $N_{K/F}$. It is well-known that every $\varphi$ in  $O(N_{K/F})$ is the composition of some automorphism in ${\rm Gal}(K/F)$ and multiplication by some element of norm one, and  every $\varphi$ in  $S(N_{K/F})$ is the composition of some automorphism in ${\rm Gal}(K/F)$ and multiplication by some non-zero element in $K$; more precisely, For every $u\in K^\times$ and $\sigma^i\in {\rm Gal}(K/F)$, the $F$-linear map $\rho:K\to K,$ $\rho(x)=u \sigma^i(x)$, is a similarity with similarity factor $N_{K/F}(u)$.

 By Hilbert 90, we also know that $N(u)=1$ if and only if there exists $b\in K^\times$ such that $u=\frac{b}{\sigma(b)}$.

\section{The construction of division algebras with right nucleus $K$}\label{sec:3}

\subsection{The algebra construction}

Let $K/F$ be a cyclic Galois field extension of degree $n$ with Galois group generated by $\sigma$ and let $R=K[t;\sigma]$.  Then every $f\in R$ is bounded. Put $R_m=\{ g\in R\,|\, \deg(g)<m \}$. Choose a monic $f\in R=K[t;\sigma]$  of degree $m$. Fix an integer $i_0\in \{0,\dots, m-1\}$ and  some $\nu\in K$. Choose an $F'$-linear bijective $\rho:K\to K$, define  $F_0=F\cap F'$ and suppose throughout  that $F/F_0$ is a field extension of finite degree. (Since we are looking for division algebras and MRD codes the assumption that $F/F_0$ has finite degree is needed in our setting, the construction obviously works in general, too.)

Let $b(t), c(t)\in R_m$ with  $b(t)=\sum_{i=0}^{m-1} b_it^i $. The multiplication defined on $R_m$ via
$$b(t)\circ c(t)= (b(t)+\nu\rho(b_{i_0})t^m )c(t)\,\,{\rm mod}_r f,$$
where the right hand side is the remainder of dividing $(b(t)+\nu\rho(b_{i_0})t^m )c(t)$ by $f$ on the right,
makes the set $R_m$ into a non-unital nonassociative ring $(R_m,\circ)$ (with the usual polynomial addition),
which we denote by  $\mathbb{S}_f(K,i_0,\nu,\rho)$.  This nonassociative ring  $\mathbb{S}_f(K,i_0,\nu,\rho)$ is an algebra over $F_0$
and when viewed as $F_0$-vector space, we know that
$${\rm dim}_{F_0} \mathbb{S}_f(K,i_0,\nu,\rho)=m[K:F_0].$$

If we take  the skew polynomial appearing on the right-hand side of the equation and write $(b(t)+\nu\rho(b_0)t^m )c(t)=qf+r$ for uniquely determined $q,r\in R$ with ${\rm deg}(r)<{\rm deg}(f)$ using right
division by $f$, then we can rewrite the multiplication of $\mathbb{S}_f(K,i_0,\nu,\rho)$ as
  $$b(t)\circ c(t)= (b(t)+\nu\rho(b_{i_0})t^m )c(t)-qf.$$

\begin{remark} $(i)$
 The algebra $\mathbb{S}_f= \mathbb{S}_f(K,i_0=0, 0, id)$ with multiplication
$$b(t)\circ c(t)=b(t) c(t)\,\,{\rm mod}_r f$$
is called a \emph{Petit algebra}, is unital,  and is studied for instance in
\cite{BP, BPhD, LS, P66}. A Petit algebra
 $\mathbb{S}_f$ is  associative if and only if $f$ is  two-sided \cite{P66}, i.e. $f$ generates a two-sided ideal in $R$. Recall that $f\in R=K[t;\sigma]$ is two-sided in $R$ if and only if $f(t)=g(t)t^s$ for some $g \in C(R)$ and some
integer $s\geq 0$  \cite[Theorem 1.1.22]{J96}.
  When $K$ is a finite field and $f$ is irreducible, then $\mathbb{S}_f$ is a finite division algebra and called a \emph{cyclic semifield}. When $K$ is a finite, a Petit division algebra
 $\mathbb{S}_f$  is isotopic to a
 \emph{Jha-Johnson semifield} (note that isotopic semifields are automatically considered as identical).
  \\ $(ii)$ For $\rho\in {\rm Gal}(K/F)$ and $\nu\in K^\times$, the  algebras  $\mathbb{S}_f(K,i_0=0,\nu,\rho)$
  with multiplication
  $$b(t)\circ c(t)=(b(t)+\nu\rho(b_0)t^m) c(t)\,\,{\rm mod}_r f$$
are studied in \cite{PT2, DT2020, Sheekey}. For $\rho\in {\rm Aut}(K)$, the  algebras  $\mathbb{S}_f(K,i_0=0,\nu,\rho)$ have not been studied so far.
 \\ $(iii)$ The algebra
  $\mathbb{S}_{t-a_0}(K,i_0=0,\nu,\rho)$ has the multiplication
$$b\circ c = (b+\nu\rho(b)t)c\: \,\,{\rm mod}_r f
= bc+\nu\rho(b)\sigma(c)t\: \,\,{\rm mod}_r f
= bc+\nu\rho(b)\sigma(c)a_0$$
for all $b,c\in K$.  When  $\rho\in S(N_{K/F})$ is a similarity of the norm $N_{K/F}$,  the algebras  $\mathbb{S}_{t-a_0}(K,\nu,\rho,0, id)$
are some of the generalisations of non-unital algebras that are isotopic to Albert's twisted fields that are studied in \cite{P15}.
In particular, if $K$ is a finite field, then
the algebras  $\mathbb{S}_{t-1}(K,i_0=0,\nu,\rho)$  are  pre-semifields that are isotopic to Albert's twisted semifields in \cite{A}.
\end{remark}

\subsection{Criteria for $\mathbb{S}_f(K,i_0,\nu,\rho)$ to be a division algebra}
Let $f\in  K[t;\sigma]$ be a monic and irreducible skew polynomial of degree $m$.
Note that we do allow $f$ to be two-sided here, i.e. $Rf$ may be a two-sided ideal in $R$.

Define the set
$$ P_{i_0} =\lbrace g_0+g_1t+\dots+g_{m-1}t^{m-1}+\nu\rho(g_{i_0})t^m\,|\, g_i\in K\rbrace \subset K[t;\sigma].$$
In order to be similar to $f$, a polynomial $g(t)=g_0+g_1t+\dots+g_{m-1}t^{m-1}+\nu\rho(g_{i_0})t^m\in  P_{i_0} $ must be irreducible of degree $m$, thus have non-zero coefficient $g_{i_0}$.
In other words, when $g$ is similar to $f$ but $g_{i_0}=0$ then we immediately know that $g\not \in P_{i_0}$.

\begin{theorem}  \label{thm:division}
 The algebra
$\mathbb{S}_f(K,i_0,\nu,\rho)$  is a division algebra over $F_0$ if and only if  $ P_{i_0} $ does not contain any polynomial similar to $f$.
\end{theorem}

\begin{proof}
Since our algebras are finite dimensional, we will show that $\mathbb{S}_f(K,i_0,\nu,\rho)$ has zero divisors if and only if  $ P_{i_0} $ contains a polynomial similar to $f$.

 Let  $\mathbb{S}_f(K,i_0,\nu,\rho)$ have zero divisors. Then
there are $b(t), c(t)\in R_m$, $b(t)=b_0+b_1t+\dots+b_{m-1}t^{m-1}$, such that
$$b(t)\circ c(t)=(b(t)+\nu \rho (b_{i_0})t^m) c(t){\rm mod}_r f =0.$$
 Thus
 there exists $g_0\in R$ of degree at most $m$, such that $(b(t)+\nu \rho (b_{i_0})t^m) c(t)=g_0(t)f(t).$
Since $f$ is irreducible of degree $m$, and $c(t)$ has degree less than $m$, $f$ must be similar to $b(t)+\nu \rho (b_{i_0})t^m \in P_{i_0}$, because of the uniqueness of an irreducible decomposition in $R$ up to order and similarity.

Conversely, let $ P_{i_0} $ contain a polynomial $b(t)+\nu \rho (b_{i_0})t^m$ similar to $f$, where $b(t)\in R$ has degree at most $m-1$, then $b(t)+\nu \rho (b_{i_0})t^m$ is irreducible as well, hence $b(t)\not=0$. Thus there exists nonzero $u\in R$  and $v\in R$ such that $(b(t)+\nu \rho (b_{i_0})t^m)u=vf$ by definition of similarity;
note that here we may assume w.o.l.o.g. that $u(t)$ has degree less than $m$ \cite[p.~14]{J96}.
It follows that $\mathbb{S}_f(K,i_0,\nu,\rho)$ has zero divisors.
\end{proof}

\begin{corollary}  \label{cor:division2}
 Let  $b(t), c(t)\in R_m$. If $b(t)+\nu \rho (b_{i_0})t^m \in  P_{i_0} $ is reducible of degree $m$,
or has degree less than $m$, then $b(t)\circ c(t)\not=0$  in $\mathbb{S}_f(K,i_0,\nu,\rho)$.
\end{corollary}

\begin{lemma}\label{le:one}
Let $g_0+g_1t+\dots+g_{m-1}t^{m-1}+t^m\in R$ be a monic polynomial of degree $m$. Then $\nu\rho(g_{i_0})g(t)$
is a degree $m$ skew polynomial in $P_{i_0}$ if and only if
$\nu\rho(g_{i_0}) =1$ and $g_{i_0}\not=0$.
\end{lemma}

\begin{proof}
Let $g(t)=g_0+g_1t+\dots+ g_{i_0}t^{i_0}+\dots +t^m\in R$.  Then $$\nu\rho(g_{i_0})g(t)=\nu\rho(g_{i_0})g_0+\nu\rho(g_{i_0})g_1t+\dots+\nu\rho(g_{i_0}) g_{i_0}t^{i_0}+\dots +\nu\rho(g_{i_0})t^m$$
is a degree $m$ skew polynomial in $P_{i_0}$ if and only if $\nu\rho(g_{i_0}) g_{i_0}=g_{i_0}$ and $g_{i_0}\not=0$ which is equivalent to
$\nu\rho(g_{i_0}) =1$ and and $g_{i_0}\not=0$.
\end{proof}

\begin{proposition} \label{thm:division2}
The algebra
$\mathbb{S}_f(K,i_0,\nu,\rho)$  is a division algebra over $F_0$ if and only if  every
 monic polynomial $g(t)=g_0+g_1t+\dots+ g_{i_0}t^{i_0}+\dots +t^m\in R$ that are similar to $f$, satisfies either that $\nu\rho(g_{i_0})\not =1$ and $g_{i_0}\not=0$, or that $g_{i_0}=0$.
\end{proposition}

\begin{proof}
 By Theorem \ref{thm:division},
$\mathbb{S}_f(K,i_0,\nu,\rho)$  is a division algebra if and only if  $ P_{i_0} $ does not contain any polynomial similar to $f$. By Lemma \ref{le:one},
 a monic polynomial $g(t)=g_0+g_1t+\dots+ g_{i_0}t^{i_0}+\dots +t^m$ similar to $f$ has a scalar multiple that lies in $P_{i_0} $ and has degree $m$,  if and only if $\nu\rho(g_{i_0}) =1$ and $g_{i_0}\not=0$.
\end{proof}

We can rephrase this as follows. Define
$$C_{f,{i_0}}=\{ f_{i_0}\,|\, s (t)=f_0+f_1t+\dots+f_{m-1}t^{m-1}+t^m \text{ is monic and similar to } f \text{ with } f_{i_0}\not=0\}$$
 to be the \emph{set of all non-zero $i_{0}$th coefficients} of monic polynomials similar to $f$.

\begin{proposition}\label{thm:normI}
The algebra
$\mathbb{S}_f(K,i_0,\nu,\rho)$  is a division algebra over $F$ if and only if
$$ g_{i_0}\nu^{-1}\rho(g_{i_0})^{-1}\not \in C_{f,{i_0}}$$
 for all $g(t)=g_0+g_1t+\dots+g_{m-1}t^{m-1}+\nu\rho(g_{i_0})t^m\in P_{i_0}$ which are similar to $f$ and  satisfy $g_{i_0}\not=0$.
\end{proposition}

\begin{proof}
 By Theorem \ref{thm:division},
$\mathbb{S}_f(K,i_0,\nu,\rho)$  is a division algebra if and only if  $ P_{i_0} $ does not contain any polynomial similar to $f$.

  Now $g(t)=g_0+g_1t+\dots+g_{m-1}t^{m-1}+\nu\rho(g_{i_0})t^m\in P_{i_0}$ is similar to $f$ if and only if  $g_{i_0}\not=0$
  and
$$\nu^{-1}\rho(g_{i_0})^{-1}g(t)=g_0\nu^{-1}\rho(g_{i_0})^{-1}+g_1\nu^{-1}\rho(g_{i_0})^{-1}t+\dots+g_{m-1} \nu^{-1}\rho(g_{i_0})^{-1}t^{m-1}+t^m$$
is similar to $f$. This is equivalent to  $g_{i_0}\not=0$ and the $i_{0}$th coefficient of $\nu^{-1}\rho(g_{i_0})^{-1}g(t)$
satisfying $ g_{i_0}\nu^{-1}\rho(g_{i_0})^{-1}\in C_{f,{i_0}}.$
\end{proof}

 When $i_0=0$ and $\rho\in {\rm Gal}(K/F)$ we obtain the following criterium for Sheekey's twisted cyclic algebra $\mathbb{S}_f(K,i_0=0,\nu,\rho)$  \cite{Sheekey}.

  \begin{corollary}\label{cor:Sh}
Let $\rho$  be a Galois automorphism of $K/F$.
 If $$ g_{0}\nu^{-1}\rho(g_{0})^{-1}\not\in C_{f,{0}}$$
 for all $g(t)=g_0+g_1t+\dots+g_{m-1}t^{m-1}+\nu\rho(g_{i_0})t^m\in P$ with $g_{i_0}\not=0$ which are similar to $f$, then $\mathbb{S}_f(K,i_0=0,\nu,\rho)$ is a division algebra.
 \end{corollary}

Suppose that $K/F'$ is a Galois field extension and define $N_{f,{i_0}}=N_{K/F'}(C_{f,{i_0}})$, i.e.
$$N_{f,{i_0}}=\{ N_{K/F'}(f_{i_0})\,|\, f_{i_0}\not=0 \text{ and } f' (t)=f_0+f_1t+\dots+f_{m-1}t^{m-1}+t^m \text{ is monic and similar to } f\}$$
 to be the \emph{set of all the norms of the nonzero $i_{0}$th coefficients of monic polynomials similar to $f$}.

 \begin{theorem}\label{thm:normstrong}
 If
 $$ N_{K/F'}(g_{i_0}\nu^{-1}\rho(g_{i_0})^{-1})\not\in N_{f,{i_0}}$$
  for all $g=g_0+g_1t+\dots+g_{m-1}t^{m-1}+\nu\rho(g_{i_0})t^m\in P$ with $g_{i_0}\not=0$ which are similar to $f$, then $\mathbb{S}_f(K,i_0,\nu,\rho)$ is a division algebra.
 \end{theorem}

 \begin{proof}
 Let $g\in P_{i_0}$,
 $g(t)=g_0+g_1t+\dots+g_{m-1}t^{m-1}+\nu\rho(g_{i_0})t^m$, and $g_{i_0}\not=0$,
 The monic polynomial $\nu^{-1}\rho(g_{i_0})^{-1}g(t)$ has
$ g_{i_0}\nu^{-1}\rho(g_{i_0})^{-1}$  as its $i_{0}$th coefficient.
Now we have
$$ g_{i_0}\nu^{-1}\rho(g_{i_0})^{-1}\in C_{f,{i_0}}\Rightarrow N_{K/F'}(g_{i_0}\nu^{-1}\rho(g_{i_0})^{-1})\in N_{f,{i_0}}.$$
\end{proof}

In particular,   when $\rho'\in S(N_{K/F'})$ is a norm  similarity with similarity factor $a\in F'^\times$, then we obtain
$$ g_{i_0}\nu^{-1}\rho'(g_{i_0})^{-1}\in C_{f,{i_0}}\Rightarrow a N_{K/F'}(\nu^{-1})\in N_{f,{i_0}}$$
which leads to the following observation.

\begin{corollary}\label{cor:Sh2}
 Let
 $\rho$  be a Galois automorphism of $K/F'$ and define $\rho'(z)=u\rho(z)$ for all $z\in K$ for some $u\in K^\times$ with $N_{K/F'}(u)=a$.
If
$$ aN_{K/F'}(\nu^{-1})\not\in N_{f,{0}}$$
then $\mathbb{S}_f(K,i_0=0,\nu,\rho')$ is a division algebra.
\end{corollary}

  Note that the algebras $\mathbb{S}_f(K,i_0=0,\nu,\rho)$ in Corollary \ref{cor:Sh2} are Sheekey's twisted semifields when $F$ is finite.  Assuming that  $\deg(h)=mn$ is maximal possible, we also already knew that if $\rho\in \Gal(K/F)$ and
  $$N_{K/F_0}(a_0 \nu)\not=1$$
   then $\mathbb{S}_f(K,i_0=0,\nu,\rho)$ is a division algebra (cf. \cite[Theorem 7]{Sheekey} for finite fields, or  \cite[Theorem 22]{PT} for any cyclic Galois field extension).

 When $f$ and $\tilde f$ are similar then the algebras $\mathbb{S}_f(K,i_0=0,0,\rho)$ and $S_{\tilde f}(K,i_0=0,0,\rho)$ are isotopic  \cite{Pum2025}. However, we do not know if or when two algebras $\mathbb{S}_f(K,i_0,\nu,\rho)$ and $S_{\tilde f}(K,i_0,\nu,\rho)$ are isotopic. This seems to be a hard question.

\subsection{Norm conditions for obtaining division algebras}

We can imitate the methods successfully employed in \cite{Sheekey, NewS} to obtain conditions for our algebras to be division algebras.
 While we are no experts we conjecture that the computational effort for finite fields seems to be similar for all methods presented in this paper.

 We know that
 $R=K[t;\sigma]$  has center
$ F[t^n]=\{\sum_{i=0}^{k}c_i(t^n)^i\,|\, c_i\in F \}= F[x]$, and put $x=t^n$, cf. Section \ref{subsec:skewpol}.

For any $f \in R$, the \emph{minimal central left multiple of $f$ in $R$} is the unique  polynomial of minimal degree $h \in C(R)=F[t^n]$,
such that $h = gf$ for some $g \in R$, and such that $h(t)=\hat{h}(t^n)$ for some monic $\hat{h}(x) \in F[x]$. Let $h_0$ be the constant term of $h(t)$.
  Since we assume that $f$ is irreducible in our constructions, we have $f^*\in C(R)$ \cite[Lemma 2.11]{GLN18} and $h(t)$ equals the bound $f^*$ of $f$ up to a scalar multiple in $K^\times$.

 From now on let  $f(t)=a_0+a_1t+\dots+ a_{m-1} t^{m-1}+t^m \in R$ be an irreducible monic irreducible polynomial of degree $m$.
 We will also assume that $\hat{h}(x)\not=x$.

    Since $f$ is irreducible in $R$,  $\hat{h}(x)$ is irreducible in $F[x]$
    and so $h$ generates a maximal two-sided ideal $Rh$ in $R$
\cite[p.~16]{J96} and $R/Rh$ is simple over $ E_{\hat{h}}$.  The quotient algebra $R/Rh$ has centre $ C(R/Rh)\cong  E_{\hat{h}}=F[x]/ (\hat{h}(x) )$ \cite[Lemma 4.2]{GLN18}.
Note  that $h(t)$ has maximal possible degree ${\rm deg}(h)=mn$, if $\gcd (m,n)=1$, or if $F$ is a finite field.

Fix some $i_0\in \{0,\dots , m-1\}$.

\begin{theorem}\label{thm:mainI}
Suppose  that $F/F_0$ is a finite Galois field extension, and $\nu\in K^\times$, $\rho$ a bijective $F$-linear map.
 Then $\mathbb{S}_f(K,i_0,\nu,\rho)$  is a division algebra over $F_0$ of dimension $mn[F:F_0]$ if one of the following holds.
\\ (i) All irreducible factors $g(t)=\sum_{i=0}^{m} g_it^i$ of $h$   satisfy
 $$ N_{K/F}(g_0 \nu^{-1}\rho(g_{i_0})^{-1} )\not=(-1)^{m(n-1)}h_0^{n/k}.$$
 (ii)  $\rho$ is an isometry of $N_{K/F}$, that is $\rho(v)=u\sigma^i(v)$ for some $i\in\{0,\dots, n-1  \}$, $u\in K$ with $N_{K/F}(u)=1$, and  all irreducible factors $g(t)=\sum_{i=0}^{m} g_it^i$  of $h$   satisfy
$$ N_{K/F}(g_0 \nu^{-1}g_{i_0}^{-1} )\not=(-1)^{m(n-1)}h_0^{n/k}.$$
 (iii) $\rho$ is an isometry of $N_{K/F_0}$,  that is $\rho(v)=u\varphi(v)$ for some $\varphi\in \Gal(K/F_0)$, $u\in K$ with $N_{K/F}(u)=1$, and all irreducible factors $g(t)=\sum_{i=0}^{m} g_it^i$  of $h$  satisfy
 $$N_{K/F_0}(g_0 \nu^{-1}\rho(g_{i_0})^{-1} )\not=(-1)^{m(n-1)[F:F_0]}N_{K/F_0}(h_0)^{n/k}.$$
\end{theorem}

\begin{proof}
$(i) $ By Theorem \ref{thm:division}, $\mathbb{S}_f(K,i_0,\nu,\rho)$
 is a division algebra, if the set $ P_{i_0} $  does not contain any polynomial similar to $f$. All polynomials similar to $f$ are
irreducible factors of $h(t)$, so $\mathbb{S}_f(K,i_0,\nu,\rho)$ is a division algebra, if $ P $ does not contain any irreducible factor of $h(t)$.
Suppose that $P_{i_0} $ contains an irreducible factor $g(t)=\sum_{i=0}^{m} g_it^i$ of $h$. Clearly,
 $g$ has degree $m$ as it is similar to $f$. Let  $g_{m} t^{m}$ be its highest coefficient, so that $g_{m}^{-1}g$ is a monic divisor of $h$.

Since  $g_{m}^{-1}g$ is a monic divisor of $h(t)$ in $R$ of degree $m$, we have
$$N_{K/F}(g_m^{-1}g_0 )=(-1)^{m(n-1)}h_0^{n/k}$$
 by \cite[Theorem 3.15]{NewS}.
Substituting $g_{m}=\nu\rho(g_{i_0})$
into this equation yields
$$N_{K/F}(g_0 \nu^{-1}\rho(g_{i_0})^{-1} )=(-1)^{m(n-1)}h_0^{n/k}$$
 $(ii)$ is trivial.
 \\ $(iii)$ Let $\rho $ be an isometry of $N_{K/F_0}$. Applying $N_{F/F_0}$ to both sides of the equation
 $$N_{K/F}(g_0 \nu^{-1}\rho(g_{i_0})^{-1} )=(-1)^{m(n-1)}h_0^{n/k}$$
  in (i) implies that
$$N_{K/F_0}(g_0 \nu^{-1}\rho(g_{i_0})^{-1} )=(-1)^{m(n-1)[F:F_0]}N_{K/F_0}(h_0)^{n/k}.$$
\end{proof}

  Similar conditions can be of course obtained for similarities $\rho\in S(N_{K/F})$.

  For $i_0=0$ and $\rho\in {\rm Gal}(K/F)$, we obtain \cite[Theorem 22]{PT} as special case of Theorem \ref{thm:mainI} $(ii)$.

\begin{corollary}\label{cor:mainI}
Suppose  that $F/F_0$ is a finite Galois field extension, and $\nu\in K^\times$, $\rho$ a bijective $F$-linear map.
 Then $\mathbb{S}_f(K,i_0,\nu,\rho)$  is a division algebra over $F_0$ of dimension $mn[F:F_0]$ if one of the following holds.
\\ (i) All irreducible factors $g(t)=\sum_{i=0}^{m} g_it^i$ of $h$   satisfy
 $$ N_{K/F}(g_0\rho(g_{i_0})^{-1} )\not= N_{K/F}(a_0  \nu).$$
 (ii)  $\rho$ is an isometry of $N_{K/F}$ and  all irreducible factors $g(t)=\sum_{i=0}^{m} g_it^i$  of $h$   satisfy
$$N_{K/F}(g_0 g_{i_0}^{-1})\not= N_{K/F}(a_0 \nu).$$
 (iii) $\rho$ is an isometry of $N_{K/F_0}$ and all irreducible factors $g(t)=\sum_{i=0}^{m} g_it^i$  of $h$  satisfy
$$N_{K/F_0}(g_0 g_{i_0}^{-1})\not= N_{K/F_0}(a_0 \nu).$$
\end{corollary}

This follows directly from Theorem \ref{thm:mainI} observing that when ${\rm deg}(h)=mn$ is maximal, then we have $N_{K/F}(a_0)=(-1)^{m(n-1)}h_0$ for the constant term $h_0$ of $h(t)$.

 Theorem \ref{thm:mainI} is obviously weaker than Theorem \ref{thm:division}. Moreover, the conditions given in both Theorem \ref{thm:division} and  Theorem \ref{thm:mainI} are not easily tractable in general. However,  over finite fields $K$ there are only finitely many monic polynomials $g$ that are similar to $f$,  and not all need to be contained in the set $P_{i_0}$. Hence there are only finitely many $g_0$ and $g_{i_0}$ to check.

\section{MRD-codes and unital division algebras}

Let $f(t)=\sum_{i=0}^ma_it^i\in R=K[t;\sigma]$ be an irreducible monic polynomial of degree $m>1$  with minimal central left multiple $h(t)=\hat{h}(t^n)\in F'[t^n]$. Let $k$ be the number of irreducible factors of $h(t)$ in $R$, i.e  ${\rm deg}(h)=km$.

Let $B={\rm Nuc}_r(\mathbb{S}_f)$ be the right nucleus of the Petit algebra $\mathbb{S}_f=R/Rf$, then
$B$ is a central division algebra over $E_{\hat{h}}$ of degree $s=n/k$, and $ R/Rh \cong M_k({\rm Nuc}_r(\mathbb{S}_f)).$
 In particular,  ${\rm deg}(\hat{h})=\frac{m}{s}$ and
 $${\rm deg}(h)=km=\frac{n m}{s}.$$
 Moreover, $s$ divides $\gcd (m,n)$. If $f$ is not two-sided, then $k>1$ and $s\not=n$ \cite[Theorem 3]{PT2}.
When ${\rm deg }(h)=mn$, then $B\cong E_{\hat{h}}=F[x]/ (\hat{h}(x) )$ is a field, and $s=1$.

The left $R$-module $R/Rf$ is a right $B$-module of rank $k$ via the scalar multiplication $R/Rf\times B\longrightarrow R/Rf$, $(a+Rf)(z+Rf)=az+Rf$
for all $z\in F[u^{-1}t^n]$ and $a\in R$. We sometimes identify $R/Rf$ with $B^{k}$ via a canonical basis.

 Let $K/F$ be a cyclic Galois field extension of degree $n$.
 We will usually  consider $F=\mathbb{F}_{q}$, $K=\mathbb{F}_{q^n}$, where we have $B=E_{\hat{h}}=F'[x]/ \hat{h}(x)\cong \mathbb{F}_{q^m}$ is a field extension so that  we will mostly work with matrix codes over finite fields.

Choose an $F'$-linear bijective $\rho:K\to K$ such that $F/F_0$  with $F_0=F\cap F'$ is  is a finite field extension.

We generalise the MRD-code constructions presented in \cite{Sheekey} and \cite{NewS}: we will employ   $F'$-linear maps $\rho$ and $i_0\in\{0,\dots,lm-1\}$ and not just $i_0=0$ and
$\rho\in {\rm Aut}(K)$.

Suppose throughout this section that $\nu\in K^\times$, $1\leq l<k$,
 and  assume that $f$ is not two-sided. Define
 $$P_{i_0,l}=\lbrace d_0+d_1t+\dots+d_{lm-1}t^{lm-1}+\nu\rho(d_{i_0})t^{lm}\,|\, d_i\in K \rbrace,$$
 and
$$\mathbb{S}_{n,m,l}(K,i_0,\nu,\rho, f)=\lbrace g(t)+Rh\,|\, g(t)\in P_{i_0,l}\rbrace\subset R/Rh.$$
Note that $P_{i_0,l}=P_{i_0}$ and  $\mathbb{S}_{n,m,1}(K,i_0,\nu,\rho, f)= \mathbb{S}_f(K,i_0,\nu,\rho)$.
We now explain how  $\mathbb{S}_{n,m,l}(K,i_0,\nu,\rho, f)$ defines an $F_0$-linear code  $\mathcal{C}_{i_0, n,m,l}$ in $M_k(B)$.

   Let $L_a:R/Rf\to R/Rf$ be the left multiplication $L_a(b+Rf)=ab+Rf$ in the Petit algebra $\mathbb{S}_f=R/Rf$. Since
$L_a$ is $B$-linear, we know that $L_a\in {\rm End}_{B}(R/Rf)$, also
 $$ R/Rh \cong M_k(B)\cong {\rm End}_{B}(B^k)={\rm End}_{B}(R/Rf),$$
  \cite{PT}.  Hence we have well-defined maps
$$L: \mathbb{S}_{n,m,l}(K,i_0,\nu,\rho, f) \to {\rm End}_B(R/Rf), a\mapsto L_a,$$
$$\lambda: \mathbb{S}_{n,m,l}(K,i_0,\nu,\rho, f) \to M_k(B), a\mapsto L_a \mapsto M_a,$$
 where $M_a$ is the matrix representing $L_a$
with respect to a $B$-basis of $R/Rf$. We denote  the image of $\mathbb{S}_{n,m,l}(K,i_0,\nu,\rho,f)$ in $M_{k}(B)$ by
$$\mathcal{C}_{i_0,n,m,l}=\lbrace M_a\mid a\in \mathbb{S}_{n,m,l}(K,i_0,\nu,\rho,f)\rbrace.$$
The code $\mathcal{C}=\mathcal{C}_{i_0, n,m,l}\subset M_{k}(B)$  is $F_0$-linear by construction, and a generalized rank metric code.
If  $d_{\mathcal{C}}$ is the minimum distance of $\mathcal{C}$ then the Singleton-like bound canonically generalizes to the bound
$${\rm dim}_{F_0}(\mathcal{C})\leq k(k-d_{\mathcal{C}}+1)[B:F_0],$$
 with $[B:F_0]=s[F:F_0]$.
If $d_{\mathcal{C}}=k-l+1$, then
$${\rm dim}_{F}(\mathbb{S}_{n,m,l}(K,\nu,\rho, f))=d^2nml/dms [B:F]=d^2mnl [F':F] =lk [B:F]=lkdms[F' :F].$$
 Thus if $d_{\mathcal{C}}=k-l+1$, then $\mathcal{C}$ attains this bound and $\mathcal{C}$ is
a  \textit{maximum rank distance code} in $M_k(B)$.
 We know $M_a\in M_{k}(B)$ and
$$ {\rm colrank}(M_a)=k-\frac{1}{m}{\rm deg}({\rm gcrd}(a(t),h(t)))$$
 for all $a+Rh\in R/Rh$, see \cite[Theorem 3.13]{NewS} or \cite[Theorem 7]{PT2}.
 If ${\rm deg}(h)=mn$ is maximal possible, then $B=E_{\hat{h}}$ is a field extension of $F$, $M_a\in M_n(E_{\hat{h}})$, and
$${\rm rank}(M_a)= n - \frac{1}{m}{\rm deg}({\rm gcrd}(a(t),h(t)).$$
 For finite fields ${\rm deg}(h)=mn$ is always true and this equality  is given in \cite[Proposition 7]{Sheekey}.

\begin{theorem} \label{thm:MRD}
 The set
$\mathbb{S}_{n,m,l}(K,i_0,\nu,\rho, f)$  defines an $F_0$-linear MRD-code $\mathcal{C}_{i_0, n,m,l}$  in $M_{k}(B)$
 with minimum distance $k-l+1$, if and only if the set $P_{i_0,l}$ does not contain any polynomial of degree $lm$ with $l$ irreducible factors which are all similar to $f$.
\end{theorem}

\begin{proof}
Recall first that  $h(t)$, ${\rm deg}(h)=km,$ has a decomposition into $k$ iredudible skew polynomials which are uniquely determined up to isometry and order, so that any greatest common right divisor of $g$ and $h$ has degree $sm$ for some integer $s$, $1\leq s\leq l$.

The set
$\mathbb{S}_{n,m,l}(K,i_0,\nu,\rho, f)$ defines an $F_0$-linear MRD-code $\mathcal{C}_{i_0, n,m,l}$
if and only if the minimum column rank of the matrix corresponding to a nonzero element in $\mathbb{S}_{n,m,l}(K,i_0,\nu,\rho, f)$ is $k-l+1$.
By \cite[Theorem 3.13]{NewS}  or \cite[Theorem 7]{PT2} we know that
$${\rm colrank}(M_a)= k - \frac{1}{m}{\rm deg}({\rm gcrd}(a(t),h(t))$$
hence this is equivalent to showing that  for all $g\in P_{i_0,l}$,
the greatest common right divisor of $g$ and $h$ has degree at most $(l-1)m$, since then
$${\rm colrank}(M_a)= k - \frac{1}{m}{\rm deg}({\rm gcrd}(a(t),h(t))\geq k-(l-1)=k-l+1.$$
Suppose towards a contradiction that there exists some $g\in P_{i_0,l}$, such that ${\rm deg}({\rm gcrd}(g,h))=lm$. Since ${\rm deg}(g)\leq lm$, it follows that $g$ must be a right divisor of $h$, and that hence $\deg (g)=lm$. (Since $g$ lies in $P_{i_0,l}$ this implies that $g_{i_0}\not=0$.)
 Every divisor of $h$ is a product of irreducible polynomials similar to $f$, hence $g$ must be a product of $l$ polynomials of degree $m$ that are all similar to $f$.
  This contradicts our assumption, so any matrix in $\mathcal{C}_{i_0, n,m,l}$ has rank at least $k-l+1$.
\end{proof}

\begin{theorem} \label{thm:alsomain}
  If all $g=g_0+g_1t+\dots+g_{m-1}t^{m-1}+\nu\rho(g_{i_0})t^{lm} \in P_{i_0,l}$
of degree $lm$ which are factors  of $h$  of degree $lm$ (i.e., products of irreducible polynomials similar to $f$) satisfy
$$N_{K/F_0}(\nu^{-1} \rho(g_{i_0})^{-1}  g_0 )\not=(-1)^{lm(n-1)}h_0^{ln/k},$$
then the set $S=\mathbb{S}_{n,m,l}(K,i_0,\nu,\rho, f)$ defines an $F_0$-linear MRD-code $\mathcal{C}_{i_0, n,m,l}$  in $M_n(E_{\hat{h}})$ with minimum distance $k-l+1$.
\end{theorem}

This generalises \cite[Theorem 5.1]{NewS}, where $\rho\in {\rm Aut}(K)$ and $i_0=0$.

\begin{proof}
  Suppose towards a contradiction that there exists $g\in P_{i_0,l}$ such that ${\rm deg}({\rm gcrd}(g,h))=lm$; since ${\rm deg}(g)\leq lm$,  $g$ must be a divisor of $h$, therefore $g$ is a product of polynomials that are all similar to $f$. Denote its highest coefficient by $g_{lm}$.
Since  $g_{lm}^{-1}g$ is a monic divisor of $h(t)$ in $R$ of degree $lm$, we have
$$N_{K/F_0}(g_{lm}^{-1}g_0 )=(-1)^{lm(n-1)}h_0^{ln/k}$$
 by \cite[Theorem 3.15]{NewS}.
 Now $g_{lm}=\nu \rho(g_{i_0})$ so
 $$N_{K/F_0}(\nu^{-1} \rho(g_{i_0})^{-1}  g_0 )=(-1)^{lm(n-1)}h_0^{ln/k}.$$
\end{proof}

When working with infinite fields, these criteria seem not easily tractable, unless we choose special $\rho$ where we can cancel terms occuring in the norm condition.

\begin{corollary}
 Fix $\nu \in K^\times$.  If one of the following is satisfied, then the set $S=\mathbb{S}_{n,m,l}(K,i_0,\nu,\rho, f)$ defines an $F_0$-linear MRD-code $\mathcal{C}_{i_0, n,m,l}$  in $M_n(E_{\hat{h}})$ with minimum distance $n-l+1$.
\\ (i) All $g=g_0+g_1t+\dots+g_{m-1}t^{m-1}+\nu\rho(g_{i_0})t^m \in P_{i_0,l}$
of degree $lm$ which are products of irreducible polynomials similar to $f$ satisfy
$$N_{K/F_0}(g_0\rho(g_{i_0})^{-1} )\not= N_{K/F_0}( \nu) N_{K/F_0}(a_0)^l.$$
\\ (ii) Let $\rho\in {\rm Gal}(K/F)$. Suppose that all $g=\sum_{i=0}^{lm} g_it^i\in P_{i_0,l}$ which are products of irreducible polynomials similar to $f$ satisfy
 $$N_{K/F}(\frac{g_0}{ g_{i_0}})\not= N_{K/F}( \nu) N_{K/F}(a_0)^l.$$
\end{corollary}

This follows directly from Theorem \ref{thm:alsomain} observing that 
we have $N_{K/F}(a_0)=(-1)^{m(n-1)}h_0^{n/k}$ for the constant term $h_0$ of $h(t)$, e.g. see \cite[Theorem 3.12]{NewS}.

For finite fields, $i_0=0$ and $\rho \in {\rm Aut}(K)$, this is \cite[Theorem 7]{Sheekey} or \cite[Theorem 24]{PT}.

\begin{example}
Let $F=\mathbb{F}_q=\mathbb{F}_3$, $K=\mathbb{F}_{q^3}=\mathbb{F}_{27}$ and $\sigma$ generate the Galois group of $K/F$. For every monic irreducible quadratic $f\in K[t;\sigma]$ we know that $h$ has degree 6. We have three classes of mutually similar monic irreducible quadratic polynomials. Each class has 117 elements. Choosing $l=1$ means we construct pre-semifields with center $\mathbb{F}_3$ and right nucleus $\mathbb{F}_9$, and  thus  associated MRD codes.

Choosing $l=2$ we construct linear MRD codes
 in $M_{3}(\mathbb{F}_9)$ with minimum distance $n-l+1=2$. Here,
 $$P_{i_0,2}=\lbrace d_0+d_1t +d_2t^2+d_{3}t^{3}+\nu\rho(d_{i_0})t^{4}\,|\, d_i\in K \rbrace.$$
  We have to check for which choices of $\nu$ and $\rho$ there are no polynomials in $P_{i_0,2}$ that are products of two irreducible polynomials that are both similar to $f$. If $g\in P_{i_0,2}$ and $g_{i_0}=0$ then $g$ is not similar to such a product by a simple degree argument (see the proof of Theorem \ref{thm:MRD}), so we only need to check the polynomials in $P_{i_0,2}$ where $g_{i_0}\not=0$.
\end{example}

\section{Isotopic division algebras}

If $\mathbb{S}_f(K,i_0,\nu,\rho)$ is a division algebra, and if we write left multiplication by $0\not=a\in \mathbb{S}_f(K,\nu,\rho, 0,id)$  as a matrix $M_a$
with coefficients in $B$, then every matrix in
$$\mathcal{C}_{i_0,n,m,1}=\lbrace M_a\in M_k(B)\mid b+Rh\in S\rbrace$$
 has full column rank and
canonically defines an $F_0$-linear MRD-code  in $M_k(B)$ with minimum distance $k$. If ${\rm deg}(h)=mn$,
 we obtain an $F_0$-linear MRD-code in $M_{n}(E_{\hat{h}})$ with minimum distance $n$.

We can sometimes
define a multiplication $*$ on $R_m$ in such a way  that the endomorphisms that represent the left multiplication in $(R,*)$ coincide with the endomorphisms in our sets $\mathbb{S}_{n,m,l}(K,i_0,\nu,\rho, f)$ when $l=1$, where we consider the image of $\mathbb{S}_{n,m,l}(K,i_0,\nu,\rho,h)$ in $M_{k}(B)={\rm End}_{B}(R/Rf)$. Under certain conditions, the algebra $(R_m,*)$ we obtain this way is isotopic to the unital Petit division algebra $\mathbb{S}_f$, viewed as algebra over $F_0$, and hence is automatically also a division algebra.

This section closely follows some ideas presented in \cite[Section 5]{NewS}.

We start by assembling three lemmata which will provide conditions for constructing non-unital division algebras $(R_m,*)$ that are isotopic to $\mathbb{S}_f$.

\begin{lemma}\label{le:important0}
Let $f(t)=a_0+a_1t+\dots+ a_{m-1} t^{m-1}+t^m \in R$ be monic, irreducible and not two-sided.
Assume that $K/F_0$ is a finite Galois field extension and that $\rho$ is $F'$-linear.
 Define the $F'$-linear map $\rho_{\nu,i_0}: K\to K$, $\rho_{\nu,i_0}(y)=y-\nu a_{i_0}\rho(y)$ and assume that $\rho_{\nu,i_0}$ is bijective.
\\ (i)  The map
$$G_{i_0}:R_m\to S_{n,m,1}(K,i_0,\nu,\rho, f)=\lbrace g(t)+Rh\,|\, g(t)\in P_{i_0}\rbrace\subset R/Rh$$
$$\sum_{i=0,i\not=i_0}^{m-1} g_it^i +\rho_{\nu,i_0}(g_{i_0})t^{i_0}\mapsto
\sum_{i=0,i\not=i_0}^{m-1} g_it^i +\rho_{\nu,i_0}(g_{i_0})t^{i_0}+\nu\rho(g_{i_0})f(t),$$
is an $F_0$-linear injective linear map and $G_{i_0}(R_m)$ is an $F_0$-sub vector space of ${\rm End}_B(R/Rf)$.
 \\ (ii) Let $F=F_0=F'$. Then every nonzero element in $G_{i_0}(R_m)$ is invertible if
 $$N_{K/F}( \rho(g_{i_0})^{-1}  g_0 )\not=N_{K/F}(a_0 \nu)$$
for all $g(t)=\sum_{i=0}^{m}g_it^i$ that are similar to $f$ with $g_{i_0}\not=0$.
\end{lemma}

Note that the condition in the setup of  $(ii)$ guarantees that $\rho_{\nu,i_0}$ is  bijective and $G_{i_0}$ is injective, see Lemma \ref{le:important} below.

\begin{proof}
$(i)$
 Since $\rho_{\nu,i_0}$ is bijective, every $g(t)\in R_m$ can be written in the form
$$\sum_{i=0,i\not=i_0}^{m-1} g_it^i +\rho_{\nu,i_0}(g_{i_0})t^{i_0}$$
for some $g_i\in K$.
 Analogously as in the proof of \cite[Lemma 5.3]{NewS}, it is easy to see that $G_{i_0}$ is a well-defined $F_0$-vector space homomorphism since we have
$$\sum_{i=0,i\not=i_0}^{m-1} g_it^i +\rho_{\nu,i_0}(g_{i_0})t^{i_0}+\nu\rho(g_{i_0})f(t)=
\sum_{i=0,i\not=i_0}^{m-1} (g_i +a_i \nu \rho(g_{i_0})  )t^i +   g_{i_0}t^{i_0}+ \nu\rho(g_{i_0})t^m \in P_{i_0}.$$
 Assume that
$\sum_{i=0,i\not=i_0}^{m-1} g_it^i +\rho_{\nu,i_0}(g_{i_0})t^{i_0}+\nu\rho(g_{i_0})f(t)$ is  zero, then we have
$$\sum_{i=0,i\not=i_0}^{m-1} g_it^i +\rho_{\nu,i_0}(g_{i_0})t^{i_0}+\nu\rho(g_{i_0})f(t)=0\, {\rm mod}_r \, f$$
which is the same as $g_i=0$ for all $i\in \{0,\dots, m-1 \}$, $i\not=i_0$, and $\rho_{\nu,i_0}(g_{i_0})=0$ which is the same as $g_{i_0}=0$.
Thus $G_{i_0}$ is injective.
\\ $(ii)$  We know that
$N_{K/F}(a_0)=(-1)^{m(n-1)}h_0^{n/k}$. Hence every nonzero element in $G_{i_0}(R/Rf)$ is invertible if
 $$N_{K/F}( \nu^{-1}\rho(g_{i_0})^{-1}  g_0 )\not=(-1)^{m(n-1)}h_0^{n/k}= N_{K/F}(a_0 )$$
for all $g$ that are similar to $f$ with $g_{i_0}\not=0$ by Theorem \ref{thm:alsomain}.
\end{proof}

The following two lemmata highlight two important special cases of this setup.

\begin{lemma}\label{le:important}
Let $f(t)=a_0+a_1t+\dots+ a_{m-1} t^{m-1}+t^m \in R$ be monic, irreducible and not two-sided.
Assume that $\rho$ is $F$-linear. Define the $F$-linear map $\rho_{\nu,i_0}: K\to K$, $\rho_{\nu,i_0}(y)=y-\nu a_{i_0}\rho(y)$ and
and the $F$-linear map
$$G_{i_0}:R_m\to S_{n,m,1}(K,i_0,\nu,\rho, f),$$
$$\sum_{i=0,i\not=i_0}^{m-1} g_it^i +\rho_{\nu,i_0}(g_{i_0})t^{i_0}\mapsto
\sum_{i=0,i\not=i_0}^{m-1} g_it^i +\rho_{\nu,i_0}(g_{i_0})t^{i_0}+\nu\rho(g_{i_0})f(t).$$
\\ (i) The map $\rho_{\nu,i_0}: K\to K$ is  bijective for all $\nu\in K^\times$ such that
$N_{K/F}(y \rho(y)^{-1})\not =N_{K/F}(\nu a_{i_0}) $ for all $y\in K^\times$.
\\ (ii)
 Let $i_0=0$ and $\nu\in K^\times$  such that $N_{K/F}(y \rho(y)^{-1})\not =N_{K/F}(a_{i_0} \nu ) $ for all $y\in K^\times$.  Then every nonzero element in $G_{0}(R_m)$ is invertible.
 \end{lemma}

\begin{proof}
$(i)$ The map is clearly $F$-linear, so it remains to show it is injective. Let $y\in K^\times$ then $\rho_{\nu,i_0}(y)=0$ is equivalent to $y=\nu a_{i_0}\rho(y)$. Apply the norm $N_{K/F}$ to both sides to get  $N_{K/F}(y)=N_{K/F}(\nu a_{i_0}) N_{K/F}(\rho(y))$ which implies
$N_{K/F}(y \rho(y)^{-1})=N_{K/F}(\nu a_{i_0}) N_{K/F'}(\rho(y))$, so if
$N_{K/F}(y \rho(y)^{-1})\not =N_{K/F}(\nu a_{i_0}) $ for all $y\in K^\times$ then it is injective.
\\ $(ii)$ We know that
$N_{K/F}(a_0)=(-1)^{m(n-1)}h_0^{n/k}$. Hence every nonzero element in $G_{0}(R/Rf)$ is invertible if
 $$N_{K/F}( \nu^{-1}\rho(g_{0})^{-1}  g_0 )\not=(-1)^{m(n-1)}h_0^{n/k}= N_{K/F}(a_0 )$$
for all $g$ that are similar to $f$ with $g_{0}\not=0$ by Theorem \ref{thm:alsomain}.

\end{proof}

We can derive an analogous result for $F'$-linear maps $\rho_{\nu,i_0}: K\to K$, where $\rho\in S(N_{K/F_0})$ with similarity factor $a\in F_0^\times$.

\begin{lemma} \label{le:important2}
Let $f(t)=a_0+a_1t+\dots+ a_{m-1} t^{m-1}+t^m \in R$ be monic, irreducible and not two-sided.
Assume that $K/F_0$ is a Galois field extension and that $\rho\in S(N_{K/F_0})$ with similarity factor $a\in F_0^\times$.
 Define the $F'$-linear map $\rho_{\nu,i_0}: K\to K$, $\rho_{\nu,i_0}(y)=y-\nu a_{i_0}\rho(y)$ and the $F_0$-linear map
$$G_{i_0}:R_m\to S_{n,m,1}(K,i_0,\nu,\rho, f),$$
$$\sum_{i=0,i\not=i_0}^{m-1} g_it^i +\rho_{\nu,i_0}(g_{i_0})t^{i_0}\mapsto
\sum_{i=0,i\not=i_0}^{m-1} g_it^i +\rho_{\nu,i_0}(g_{i_0})t^{i_0}+\nu\rho(g_{i_0})f(t).$$
\\ (i) The map $\rho_{\nu,i_0}: K\to K$, $\rho_\nu(y)=y-\nu a_{i_0}\rho(y)$ is  bijective for all $\nu\in K^\times$ such that $ N_{K/F_0}(\nu a_{i_0}) \not=1/a$.
\\ (ii)
 Let $\nu\in K^\times$ such that $ N_{K/F'}(\nu a_{i_0}) \not=1/a$.  Then every nonzero element in $G_{i_0}(R/Rf)$ is invertible if
 $$ N_{K/F_0}( g_0 g_{i_0}^{-1} )\not= a N_{K/F}(a_0 \nu)$$
for all $g$ that are similar to $f$ with $g_{i_0}\not=0$.
\end{lemma}

\begin{proof}
$(i)$ It remains to show $\rho_{\nu,i_0}$ is injective. Let $y\in K^\times$ then $\rho_{\nu,i_0}(y)=0$ is equivalent to $y=\nu a_{i_0}\rho(y)$. Apply the norm $N_{K/F_0}$ to both sides to get $N_{K/F_0}(y)=N_{K/F_0}(\nu a_{i_0}) aN_{K/F_0}(y)$ which implies
$1= a N_{K/F_0}(\nu a_{i_0}) $.
\\ $(ii)$
Every nonzero element in $G_{i_0}(R/Rf)$ is invertible if
$$N_{K/F_0}(\nu^{-1} \rho(g_{i_0})^{-1}  g_0 )\not=(-1)^{m(n-1)}h_0^{n/k}$$
for all $g$ that are similar to $f$ with $g_{i_0}\not=0$ by Theorem \ref{thm:alsomain}. Since
$N_{K/F}(a_0)=(-1)^{m(n-1)}h_0^{n/k}$ this is the same as
$$ N_{K/F_0}( g_0 g_{i_0}^{-1} )\not=  aN_{K/F}(a_0 \nu)$$
for all $g$ that are similar to $f$ with $g_{i_0}\not=0$.
\end{proof}

The Lemmata \ref{le:important0}, \ref{le:important} and \ref{le:important2}  generalise \cite[Lemma 5.2, Lemma 5.3]{NewS}.

If $\rho_{\nu,i_0}$ is bijective then every element $g(t)\in R_m$ can be written as
$$ g(t)=\sum_{i=0,i\not=i_0}^{m-1} g_it^i +\rho_{\nu,i_0}(g_{i_0})t^{i_0}$$
for some suitable $g_{i_0}$. More precisely, if $\rho_{\nu,i_0}$ is bijective and $F'$-linear, then there is a bijective $F_0$-linear map
$H: R_m\to R_m$ defined via
$$H\big(\sum_{i=0,i\not=i_0}^{m-1} g_it^i +\rho_{\nu,i_0}(g_{i_0})t^{i_0}\big) = \sum_{i=0}^{m-1} g_it^i .$$
 We will use the map $H$ in the next result where we employ Lemma \ref{le:important0}. Similar results can be obtained with analogous proofs using the setups of Lemma \ref{le:important} and Lemma \ref{le:important2}, which are more restricted, but have easier conditions.

\begin{theorem}\label{thm:algebra2}
Let $f(t)=a_0+a_1t+\dots+ a_{m-1} t^{m-1}+t^m \in R$ be  irreducible  and not two-sided. Let
$$\rho_{\nu,i_0}(y)=y-\nu a_{i_0}\rho(y)$$
 be a bijective $F'$-linear map.
  Let $ g(t)=\sum_{i=0}^{m-1} g_it^i $ and $b(t)=\sum_{i=0}^{m-1} b_it^i$.
 Define
$$g(t) * b(t) =G_{i_0}(H^{-1}(g(t))) b(t){\rm mod}_r\,f$$
$$=\big(\sum_{i=0,i\not=i_0}^{m-1} g_it^i +\rho_{\nu,i_0}(g_{i_0})t^{i_0} + \nu\rho(g_{i_0})f(t)\big) \big( \sum_{i=0}^{m-1} b_it^i\big) {\rm mod}_r\,f,$$
 then $(R_m,*)$ is a non-unital division algebra  of dimension $mn[F:F_0]$ over $F_0$
 and is isotopic to the unital Petit division algebra $\mathbb{S}_f$, viewed as an algebra over $F_0$.
\end{theorem}

This generalizes Theorem \cite[Theorem 5.4]{NewS} which treats the case that $\rho\in {\rm Gal}(K/F')$ and $i_0=0$. Moreover, it shows that the algebras studied in Theorem \cite[Theorem 5.4]{NewS} are isotopic to Petit algebras.

\begin{proof}
Since $\rho_{\nu,i_0}$ is a bijective $F'$-linear map,  $H$ is a bijective $F_0$-linear map.
By Lemma \ref{le:important0},  $G_{i_0}$ is injective and therefore $G_{i_0}:R_m\to G_{i_0}(R_m)$  is a bijective $F_0$-linear map. We  identify the vector spaces $G_{i_0}(R_m)$ and $R_m$. By abuse of notation we thus obtain a bijective $F_0$-linear map $G_{i_0}:R_m\to R_m$, by restricting the codomain to the image of $G_{i_0}$ and identifying $R_m$ with this image under $G_{i_0}$.
    Therefore the algebra $(R_m,*)$ is isotopic to the Petit algebra $\mathbb{S}_f$, both viewed as algebras over $F_0$, as the multiplication of $(R_m,*)$ is defined via
  $$g(t) * b(t) =G_{i_0}( H^{-1}(g(t))) b(t) {\rm mod}_r\,f$$
  and $H$ and $G_{i_0}$ are $F_0$-linear.
  Since $f$ is irreducible, $\mathbb{S}_f$ is a division algebra over $F$, hence over $F_0$. The algebra $(R_m,*)$ is isotopic to $\mathbb{S}_f$ by construction, hence also a division algebra and has dimension $nm[F:F_0]$:
$R_m$ is an $F$-vector space of dimension $nm$, hence an $F_0$-vector space of dimension $nm[F:F_0]$.
\end{proof}

Note that $(R_m,*)$ does not have a unit element, so its nuclei need not be the same as the ones of any unital algebra isotopic to it.
If we want to construct a unital division algebra that is isotopic to $(R,*)$ we can use Kaplanski's trick (Section \ref{subsec:1}).

\begin{corollary}\label{thm:algebra1}
Let $f\in R$ be monic irreducible of degree $m$ and not two-sided.
 Suppose that  $$N_{K/F}(\nu^{-1} \rho(g_{i_0})^{-1}  g_0 )\not= N_{K/F}(a_0 \nu),$$
for all $g$ that are similar to $f$ with $g_{i_0}\not=0$.
then $(R_m,*)$  is a non-unital division algebra over $F$ of dimension $mn$
and isotopic to the unital division algebra $\mathbb{S}_f$.
\end{corollary}

\begin{proof}
If  $N_{K/F}(\nu^{-1} \rho(g_{i_0})^{-1}  g_0 )\not= N_{K/F}(a_0 \nu)$
for all $g$ that are similar to $f$ with $g_{i_0}\not=0$, then $\rho_{\nu,i_0}(y)=y-\nu a_{i_0}\rho(y)$
 is a bijective $F$-linear map by Lemma \ref{le:important} $(i)$, and (again abusing notation)
  $G_{i_0}:R_m \to R_m$ is a bijective $F$-linear map. The assumption now immediately follows from Theorem \ref{thm:algebra2}.
\end{proof}

\begin{corollary} \label{cor: importantcor}
Let $f\in R$ be monic  and irreducible of degree $m$ and not two-sided with nonzero constant term $a_0\in K^\times$. Let $\rho$ be an $F$-linear map. Let $i_0=0$ and $\nu\in K^\times$  such that $$N_{K/F}(y \rho(y)^{-1})\not =N_{K/F}(a_{0} \nu ) $$
for all $y\in K^\times$.
 Then $(R_m,*)$ is a non-unital division algebra over $F$ of dimension $mn$ and isotopic to $\mathbb{S}_f$.
\end{corollary}

As observed in \cite[p.~27]{NewS}, there exists $z(t^n)\in C(R)$ such that $z(t^n)t^{mk}=1\, {\rm mod}_r f$, so that $t^{-1}=z(t^n)t^{mk-1} \, {\rm mod}_r f$ in $R/Rf$. The following result generalizes \cite[Theorem 5.7]{NewS}, but also shows that the algebra constructed in \cite[Theorem 5.7]{NewS} is not unital.

\begin{theorem} \label{thm:algebra3}
Let $f\in R$ be monic and irreducible of degree $m$ and not two-sided. Let
$$\rho_{\nu,i_0}(y)=y-\nu a_{i_0}\rho(y)$$
 be a bijective $F'$-linear map.
 Let $z(t^n)\in C(R)$ satisfy that
 $$z(t^n)t^{mk}=1\,{\rm mod}\,f.$$
 Then $R_m$ becomes an $F_0$-algebra  of dimension $mn[F:F_0]$  via the multiplication
$$g(t) \star b(t)=G_{i_0}(H^{-1}(g(t))) \big( z(t^n)t^{mk-1} b(t)\big) \, {\rm mod}_r\, f$$
$$=\big(\sum_{i=0,i\not=i_0}^{m-1} g_it^i +\rho_{\nu,i_0}(g_{i_0})t^{i_0} + \nu\rho(g_{i_0})f(t)\big) z(t^n)t^{mk-1} \sum_{i=0}^{m-1} b_it^i \, {\rm mod}_r\,f$$
when $b(t)=\sum_{i=0}^{m-1} b_it^i$ and $g(t)=\sum_{i=0}^{m-1} g_it^i$.
This algebra is a division algebra and is isotopic to the $F_0$-algebra from Theorem  \ref{thm:algebra2} and to $\mathbb{S}_f$ viewed as an $F_0$-algebra.
\\ (ii) The algebra $(R_m,\star)$  has $t$ as a left unit element, but not as right unit element, and center $F_0$.
\end{theorem}

\begin{proof}
$(i)$ We have
$G_{i_0}(\sum_{i=0,i\not=i_0}^{m-1} g_it^i +\rho_{\nu,i_0}(g_{i_0})t^{i_0})$
$$=\sum_{i=0,i\not=i_0}^{m-1} g_it^i +\rho_{\nu,i_0}(g_{i_0})t^{i_0}+\nu\rho(g_{i_0})f(t)=
\sum_{i=0,i\not=i_0}^{m-1} (g_i +a_i \nu \rho(g_{i_0})  )t^i +   g_{i_0}t^{i_0}+ \nu\rho(g_{i_0})t^m.$$
 Put $H(b(t))= z(t^n)t^{mk-1}b(t)$ then
 $$g(t) \star b(t)=G_{i_0}(H^{-1}(g(t))) F(b(t)) \, {\rm mod}_r\, f$$
 hence $(R_m,\star)$ is isotopic to both $\mathbb{S}_f$ and the algebra from  Theorem  \ref{thm:algebra2}. In particular, it must be a division algebra.
 \\ $(ii)$ Since  we have
 $$g(t) \star t=G_{i_0}(H^{-1}(g(t))) \big( z(t^n)t^{mk}\big) \, {\rm mod}_r\, f$$
$$=\big(\sum_{i=0,i\not=i_0}^{m-1} g_it^i +\rho_{\nu,i_0}(g_{i_0})t^{i_0} + \nu\rho(g_{i_0})f(t)\big) z(t^n)t^{mk} \, {\rm mod}_r\,f$$
$$=\big(\sum_{i=0,i\not=i_0}^{m-1} g_it^i +\rho_{\nu,i_0}(g_{i_0})t^{i_0} + \nu\rho(g_{i_0})f(t)\big)  (s(t)f(t)+1) \, {\rm mod}_r\,f$$
$$=\big(\sum_{i=0,i\not=i_0}^{m-1} g_it^i +\rho_{\nu,i_0}(g_{i_0})t^{i_0} + \nu\rho(g_{i_0})f(t)\big)   \, {\rm mod}_r\,f$$
$$=\big(\sum_{i=0,i\not=i_0}^{m-1} g_it^i +\rho_{\nu,i_0}(g_{i_0})t^{i_0}\big)   \, {\rm mod}_r\,f=g(t)$$
the algebra has a right unit element given by $t$.
Note that for all $i_0\not=1$ we  get
$$ t \star b(t)
=t   z(t^n)t^{mk-1} b(t) {\rm mod}_r\,f=  z(t^n)t^{mk} b(t)\, {\rm mod}_r\,f.$$
Now $ z(t^n)t^{mk}=s(t)f(t)+1$ implies that
$$ t \star b(t)=  (s(t)f(t)+1) b(t) \, {\rm mod}_r\,f$$
$$= (s(t)f(t) b(t) +b(t))\, {\rm mod}_r\,f$$
$$= s(t)f(t) b(t)\, {\rm mod}_r\,f +   b(t)\, {\rm mod}_r\,f.$$
and unless $f$ and $b$ commmute for all $b$, this is generally not equal to $ b(t) \,{\rm mod}_r\,f$, so that $t$ is not a left unit element.
This only is the case when $f\in Z(R)$.

For $i_0=1$ we  get
$$ t \star b(t)
=\big(\rho_{\nu,i_0}(1)t^{i_0} + \nu\rho(1)f(t)\big)  z(t^n)t^{mk-1} b(t) {\rm mod}_r\,f$$
$$= \big(\rho_{\nu,i_0}(1)t^{i_0} (s(t)f(t)+1)  b(t)
+ \nu\rho(1)f(t)  (s(t)f(t)+1)  b(t) \, {\rm mod}_r\,f$$
this is generally not equal to $ b(t) \,{\rm mod}_r\,f$, so that $t$ is not a left unit element.

The algebra $(R_m,\star)$ has center $F_0$ by construction.
\end{proof}

\begin{remark}
In Theorem \ref{thm:algebra3}, we get for $i_0=0$ and $b(t)=\sum_{i=0}^{m-1} b_it^i $:
$$\big(\rho_{\nu}(g_{0})+\sum_{i=1}^{m-1} g_it^i  \big)\star b(t)$$
$$=\big(  g_{0}+
\sum_{i=1}^{m-1} (g_i +a_i \nu \rho(g_{0})  )t^i +   \nu\rho(g_{0})t^m \big) z(t^n)t^{mk-1} b(t) \,{\rm mod}_r\,f.$$
This means
$$ t \star b(t)
=t   z(t^n)t^{mk-1} b(t) {\rm mod}_r\,f=  z(t^n)t^{mk} b(t)\, {\rm mod}_r\,f.$$
Now $ z(t^n)t^{mk}=s(t)f(t)+1$ so
$$ t \star b(t)=  (s(t)f(t)+1) b(t) \, {\rm mod}_r\,f$$
$$= (s(t)f(t) b(t) +b(t))\, {\rm mod}_r\,f$$
$$= s(t)f(t) b(t)\, {\rm mod}_r\,f +   b(t)\, {\rm mod}_r\,f.$$
and unless $f$ and $b$ commmute, this is generally not equal to $ b(t) \,{\rm mod}_r\,f$.
\end{remark}

\section{The idealisers and the nuclei}\label{sec:nuc}

Let $\mathcal{C}$ be a linear matrix code in $M_k(B)$, $B$ a division algebra.
The \textit{left} and \textit{right idealisers} of $\mathcal{C}$ are defined as
$$I_l(\mathcal{C})=\lbrace A\in M_k(B)\,|\, A\mathcal{C}\subseteq\mathcal{C}\rbrace, \text{ respectively, }
I_r(\mathcal{C})=\lbrace A\in M_k(B)\,|\, \mathcal{C}A\subseteq\mathcal{C}\rbrace,$$
and the \textit{centraliser} of $\mathcal{C}$ is given by
${\rm Cent}(\mathcal{C})=\lbrace A\in M_k(B)\,|\, A C=C A\text{ for all } C\in \mathcal{C}\rbrace.$ We call
$Z(\mathcal{C})=I_l(\mathcal{C})  \cap{\rm Cent}(\mathcal{C})$ the \textit{center} of $\mathcal{C}$.

Let $A$ be an algebra over $F$, $N={\rm Nuc}_r(A)$, and
$\mathcal{C}=\mathcal{C}(A)=\lbrace L_a \,|\,a\in A\rbrace\subseteq {\rm End}_{N}(A)$
 be the spread set of  $A$, where $L_a$ is the left multiplication map in $A$. Here, the left and right idealisers of $\mathcal{C}$ can also be described as
$$I_l(\mathcal{C})=\lbrace \Phi\in {\rm End}_F(A)\,|\, \Phi\mathcal{C}\subseteq\mathcal{C}\rbrace, \text{ respectively, }
I_r(\mathcal{C})=\lbrace \Phi\in {\rm End}_F(A)\,|\, \mathcal{C}\Phi\subseteq\mathcal{C}\rbrace,$$
and the centraliser  by
${\rm Cent}(\mathcal{C})=\lbrace \Phi\in {\rm End}_{F}(A)\,|\, \Phi C=C\Phi \text{ for all } C\in \mathcal{C}\rbrace.$

\begin{theorem}(cf. \cite[Theorem 27]{PT}) 
 \label{Nucleus of Sheekey algebras}
Let $A$ be a unital division algebra. Let  $\mathcal{C}$ be the spread set of $A$ and $\mathcal{C}^{op}$ be the the spread set associated to the opposite algebra $A^{op}$.
Then
$${\rm Nuc}_l(A)\cong I_l(\mathcal{C}),\quad {\rm Nuc}_m(A)\cong I_r(\mathcal{C}),\quad
{\rm Nuc}_r(A)\cong {\rm Cent}(\mathcal{C}^{op}),\quad C(A)\cong Z(\mathcal{C}).$$
\end{theorem}

Since ${\rm Nuc}_l(A)$, ${\rm Nuc}_m(A)$ and ${\rm Nuc}_r(A)$ are unital associative subalgebras of a unital algebra $A$, the sets
$I_l(\mathcal{C})$, $I_r(\mathcal{C})$ and ${\rm Cent}(\mathcal{C}^{op})$ contain the unit matrix.

Let  $f=\sum_{i=0}^ma_it^i\in R=K[t;\sigma]$ be an irreducible monic polynomial of degree $m$  with minimal central left multiple $h$.
 If we define two linear codes $\mathcal{C}$ and $\mathcal{C}'$ in $M_n( E_{\hat{h}})$ as \emph{equivalent}, if there exists some automorphism $\varphi$ of $M_n( E_{\hat{h}})$
and invertible $X,Y\in M_n( E_{\hat{h}}),$ such that $\mathcal{C}'= X \mathcal{C}^\varphi Y$, then the sets $I_l(\mathcal{C})$ and $I_r(\mathcal{C})$ are invariants for equivalent linear rank metric codes. However, note that for the center ${\rm Cent}(\mathcal{C})$ and centraliser $Z(\mathcal{C})$ to
be code invariants as well, we need to assume that the identity is contained in each code, e.g. \cite[Proposition 4.4]{NewS}.

We look at the case that $i_0=0$. In this case large parts of the proof of \cite[Proposition 5.5]{NewS} can be adjusted to the case that $\rho$ is an $F'$-linear bijective map. For general $i_0$ these invariants are  much more complicated to compute. We look at some partial results we achieved over finite base fields in the last section.

\begin{theorem}\label{theorem 9IIa}
Let $k>1$ and ${\rm deg}(h)=km$  (i.e., $f$ is not two-sided) 
and $\nu\in K^\times$. Suppose $1\leq l\leq k/2$, $n>k$ and $lm>2$. 
Let $\mathcal{C}=\mathcal{C}_{0,n,m,l}$.
\\ (i) $I_l(\mathcal{C})\cong \lbrace b\in K \,|\, \rho(b a) = b\rho(a) \text{ for all } a\in K\rbrace$, i.e. $F'\subset I_l(\mathcal{C}).$
\\ (ii) $I_r(\mathcal{C}) \cong \lbrace b\in K\,|\, \rho(b a) = \sigma^{lm}(b)\rho(a) \text{ for all } a\in K\rbrace.$
\\ (iii) ${\rm Cent}(\mathcal{C})\cong E_{\hat{h}}=F[x]/ (\hat{h}(x) )$.
 \\ (iv) $Z(\mathcal{C})\cong \lbrace b\in F \,|\, \rho(b a) = b\rho(a)  \text{ for all } a\in K\rbrace$, i.e. $F'\subset Z(\mathcal{C}).$
\end{theorem}

\begin{proof}
Most steps in the proof of \cite[Proposition 5.5]{NewS} go through verbatim.
We identify each element in $\mathcal{C}$ with the element $g\in  S=\mathbb{S}_{n,m,l}(K,i_0,\nu,\rho, f)$ that induces it.
 \\ $(i)$ It can be shown verbatim as in  the proof of \cite[Proposition 5.5]{NewS} that
$$ I_l(\mathcal{C})\subset  \lbrace   g \text{ mod }h\,|\, deg(g)\leq ml\rbrace.$$
To check that there are no elements $g\in I_l(\mathcal{C})$ of degree higher than $lm$, we follow the proof of \cite[Proposition 5.5]{NewS} with some slight adjustments. Let $h_i$ be the $i$th coefficient of  $h(t)=t^{km}+\dots+h_0$.
We now show that
$$\lbrace g\in I_l(\mathcal{C}) \,|\, \text{deg}(g)\leq lm\rbrace= \lbrace g_0\in K \,|\, \rho(g_0 a) = g_0\rho(a)  \text{ for all } a\in K\rbrace.$$
Let $g\in I_l(\mathcal{C})$ then $gct^s  \text{ mod } h(t)
 \in \mathcal{C}$ for all $c\in \lbrace g_0\in \mathbb{F}_{q^n} \,|\, \rho(g_0 a) = g_0\rho(a)  \text{ for all } a\in \mathbb{F}_{q^n}\rbrace$ and $s\in \{1,\dots , {ml-1}\}$. This later set is a non-empty set as $lm>1.$
 Consider $s=1$; then 
 $$gt \text{ mod } h(t)=\left(\sum_{i=0}^{mn-1}g_{i-1}t^i\right)-g_{km-1}\left(\sum_{j=0}^{m-1}h_jt^{nj}\right).$$
As $g\in I_l(\mathcal{C})$, this implies $gt \text{ mod } h\in \mathcal{C}$, so for all $i\in\lbrace lm+1,\dots, mk-1\rbrace$, we have
\begin{equation}\label{eqn:nucleus g_i}
g_{i-1}=\begin{cases}
0& \text{for } i\not\equiv 0\text{ mod }n\\
g_{km-1}h_{i/n}& \text{for } i\equiv 0\text{ mod }n
\end{cases}
\end{equation}
where we put $h_{i/n}=0$ if $i/n$ is not an integer. We want to show that $g_{km-1}=0$ and thus $\text{deg}(g)\leq lm-1$. As $lm>2$, this follows verbatim as in the proof of \cite[Proposition 5.5]{NewS}.
\\ $(ii)$
For $I_r(\mathcal{C})$  we obtain that
$$I_r(\mathcal{C})\subset  \lbrace   g \text{ mod }h\,|\, deg(g)\leq ml\rbrace$$
and that  every $g\in I_r(\mathcal{C})$   is constant, that is $g=g_0$, as in $(i)$.
Since $(c_0+\nu\rho(c_0)t^{lm}) g_0 \in \mathcal{C}$ if and only if $\nu\rho(c_0) \sigma^{lm}(g_0)=\nu\rho(c_0g_0)$, which is equivalent to $\rho(c_0g_0)= \sigma^{lm}(g_0) \rho(c_0) $, we have
$$\lbrace g\in I_r(\mathcal{C}) \,|\, \text{deg}(g)\leq lm\rbrace= \lbrace g_0\in K \,|\, \rho(g_0 c_0) = \sigma^{lm}(g_0)\rho(c_0)
\text{ for all } c_0\in K \rbrace.$$
\\ $(iii)$ and $(iv)$ follow verbatim  as in the proof of \cite[Theorem 9]{Sheekey}.
They do not rely on $\rho$ being a field automorphism or on $k=m$.
\end{proof}

\begin{corollary}\label{cor:9II}
Let $k>1$ and ${\rm deg}(h)=km$  (i.e., $f$ is not two-sided) and $\nu\not=0$. Suppose $1\leq l\leq k/2$, $n>k$ and $lm>2$.
Let $\mathcal{C}=\mathcal{C}_{0,n,m,l}$.
Let $\rho:K\to K$ be an $F$-linear bijective map (i.e., we choose $F=F'=F_0$), $\rho(y)= \sum_{i=0}^{n-1}y_i\sigma^i(y)$
 for some $y_i\in K$. Put $I_\rho=\{i\in \{0,\dots, n-1\}\,| y_i\not=0\, \}$.
\\ (i) $I_l(\mathcal{C})\cong \bigcap_{i\in I_\rho} {\rm Fix}(\sigma^i)$.
\\ (ii) $I_r(\mathcal{C}) \cong \bigcap_{i\in I_\rho} {\rm Fix}(\sigma^{lm-i})$.
\\ (iii)  $Z(\mathcal{C})\cong F$.
\end{corollary}

\begin{proof}
$(i)$ We observe that $\rho(b a) = b\rho(a)$  is equivalent to
$$ \sum_{i=0}^{n-1}y_i\sigma^i(ba)= \sum_{i=0}^{n-1} by_i\sigma^i(a)$$
which is equivalent to
$$ \sum_{i=0}^{n-1}y_i (\sigma^i(b)-b) \sigma^i(a)= 0.$$
By Galois Theory, the $\sigma^0,\sigma^1,\dots, \sigma^{n-1}$ are linearly independent over $K$, hence if $\rho(b a) = b\rho(a)$ for all $a\in K$ then for all $i\in \{0,\dots, n-1\}$ we have $y_i=0$ or $b\in {\rm Fix}(\sigma^i)$. This yields the assertion.
\\ $(ii)$ is proved as $(i)$ and $(iii)$ follows from $(i)$.
\end{proof}

If exactly one $y_i$ in $\rho(x)= \sum_{i=0}^{n-1}y_i\sigma^i(x)$ is nonzero, this means that $I_l(\mathcal{C})\cong  {\rm Fix}(\sigma^i)$ and $I_r(\mathcal{C})\cong  {\rm Fix}(\sigma^{lm-i})$ as observed in \cite[Proposition 5.5]{NewS}.

If $\gcd (n,i_1,\dots,i_s)=1$ for some $i_1,\dots,i_s\in I$ then ${\rm Gal}(K/F)=\langle \sigma^{i_1},\dots, \sigma^{i_s}\rangle$ and $I_l(\mathcal{C})\cong F$.

 Recall $\rho$ is bijective  if and only if
 $N_{A_0}(y_0 + y_1 t + \cdots + y_{n-1}t^{n-1})\not=0.$

We conjecture that we can construct new classes of codes when $\rho$ is not a Galois automorphism, the next example gives an indication for this conjecture to be true.

 \begin{example}
Assume that $n=12$, $l=2$, $m=3$ and that we can find a monic irreducible skew polynomial $f\in K[t;\sigma]$ that is not two-sided to construct the matrix code $\mathcal{C}_{i_0,12,3,2}$.

Let $\rho(x) = y_0 x + y_3 \sigma^3(x)$ with $y_0, y_3 \neq 0$. Then we have $I_l(\mathcal{C})\cong {\rm Fix}(\sigma^3)$ is an intermediate field extensions of $K/F$  of degree 3 and  $I_r(\mathcal{C})\cong {\rm Fix}(\sigma^{6})$ an intermediate field extensions of $K/F$ of degree 6.

 In order to obtain left nuclei of the same degree by choosing $\rho=\sigma^i\in {\rm Gal}(K/F)$, we only have to check $i=3$ and $i=9$, as we know that
when $\rho=\sigma^i$, we get $I_l(\mathcal{C})\cong {\rm Fix}(\sigma^i)$ has degree $\gcd (12,i)$ and $I_r(\mathcal{C})\cong {\rm Fix}(\sigma^{6-i})$ has degree $\\gcd (12,6-i)$.

 So for $i=3$, $I_l(\mathcal{C})\cong {\rm Fix}(\sigma^3)\cong I_r(\mathcal{C})$ are both intermediate field extensions of $K/F$ of degree $3$.

 For $i=9$, $I_l(\mathcal{C})\cong {\rm Fix}(\sigma^9)$  and
 $I_r(\mathcal{C})\cong {\rm Fix}(\sigma^{6-i})$ are both intermediate field extensions of $K/F$ of degree 3 as well.
 It is equally straightforward to check that  the degrees of $I_l(\mathcal{C})\cong {\rm Fix}(\sigma^9)$  and
 $I_r(\mathcal{C})\cong {\rm Fix}(\sigma^{6-i})$ cannot be 6 and 3, respectively, when choosing $\rho=\sigma^i$ for some $i$.
\end{example}

\subsection{When $i_0=0$ and $F$ is a finite field}
Let $F=\mathbb{F}_{q}$, $K=\mathbb{F}_{q^n}$, and let $f(t)=\sum_{i=0}^ma_it^i\in R=\mathbb{F}_{q^n}[t;\sigma]$ be an irreducible monic polynomial of degree $m$ that is not two-sided with minimal central left multiple $h$, so that $E_{\hat{h}}\cong \mathbb{F}_{q^m}$ and $deg(h)=mn$.

In this section, we will only consider $F$-linear maps (i.e. assume $F'=F=F_0$).

We obtain Sheekey's MRD-code construction when $i_0=0$ and $\rho=\sigma^i\in {\rm Gal}(K/F)$ is a Galois automorphism. However, we will employ any bijective $F$-linear map $\rho:K\to K$,
 $$\rho(x)= \sum_{i=0}^{n-1}y_i\sigma^i(x)$$
  where the $y_i\in K$ are  chosen so that $\rho$ is bijective. Put $I_\rho=\{i\in \{0,\dots, n-1\}\,| y_i\not=0\, \}$.

Let $S=\mathbb{S}_{n,m,l}(\mathbb{F}_{q^n},i_0=0,\nu,\rho,f)$ with $\nu\neq 0$ and let $\mathcal{C}=\mathcal{C}_{i_0=0   , n,m,l}$ be the image of $S$ in  ${\rm End}_{E_f}(R/Rf)$, so that the corresponding rank metric code $\mathcal{C}$ lies in $M_n(\mathbb{F}_{q^m})$.

For finite fields, Corollary \ref{cor:9II} yields the next result, observing that for finite fields and $f$ not two-sided, we know that $k=m$.

\begin{corollary}\label{cor:theorem 9II}
Let $R=\mathbb{F}_{q^n}[t;\sigma]$ and $\nu\not=0$. Assume that $f$ is not two-sided
and that $\nu\not=0$. Suppose $1\leq l\leq m/2$, $n>m$ and $lm>2$.
 Let $\rho(y)= \sum_{i=0}^{n-1}y_i\sigma^i(y)$
 for some $y_i\in K$.  Let $\mathcal{C}=\mathcal{C}_{0,n,m,l}$.
\\ (i) $I_l(\mathcal{C})\cong \bigcap_{i\in I_\rho} {\rm Fix}(\sigma^i)$.
\\ (ii) $I_r(\mathcal{C})  \cong \bigcap_{i\in I_\rho} {\rm Fix}(\sigma^{lm-i})$.
\\ (iii) ${\rm Cent}(\mathcal{C})\cong \mathbb{F}_{q^m}$, $Z(\mathcal{C})\cong \mathbb{F}_q$.
\end{corollary}

 It is clear that the choice of $\rho$ influences the size of the idealizers $I_l(\mathcal{C}) $, $I_r(\mathcal{C}) $ and the center $Z(\mathcal{C})$.
Since $I_l(\mathcal{C}) $ and $I_r(\mathcal{C}) $
are invariants for equivalent linear codes which contain the unit matrix, if we want to find new codes other than the ones constructed in \cite{Sheekey}, we must choose $\rho$ such that $I_l(\mathcal{C})$ and $I_r(\mathcal{C})$
 are ``new''.

\begin{corollary}\label{cor:nuclei K}
Let $f\in R=\mathbb{F}_{q^n}[t;\sigma]$ be monic, irreducible and not two-sided, $\nu\neq 0$. 
 Let $\rho(y)= \sum_{i=0}^{n-1}y_i\sigma^i(y)$
 for some $y_i\in K$.
  Suppose that $n>1$, $m>2$, $l=1$, and that $S=S(\mathbb{F}_{q^n},i_0=0,\nu,\rho,f)$ is a division algebra. Then
\\ (i) ${\rm Nuc}_l(S)= \bigcap_{i\in I_\rho} {\rm Fix}(\sigma^i)$,
\\ (ii) ${\rm Nuc}_m(S)= \bigcap_{i\in I_\rho} {\rm Fix}(\sigma^{lm-i})$.
\\ (iii) $ C(S)= \mathbb{F}_{q}$.
\\ (iv) ${\rm Nuc}_r(S)=\mathbb{F}_{q^m}$.
\end{corollary}

This also follows analogously as in the proof of \cite[Theorem 9]{Sheekey}.

\subsection{When we have any $i_0$ and $F$ is a finite field}

We now consider any $i_0$ and obtain some partial results. Let $F=\mathbb{F}_{q}$, $K=\mathbb{F}_{q^n}$, and let $f(t)=\sum_{i=0}^ma_it^i\in R=\mathbb{F}_{q^n}[t;\sigma]$ be an irreducible monic polynomial of degree $m$ that is not two-sided with minimal central left multiple $h$, so that $E_{\hat{h}}\cong \mathbb{F}_{q^m}$ and $deg(h)=mn$.

\begin{theorem}\label{theorem 9II}
Let $R=\mathbb{F}_{q^n}[t;\sigma]$ and 
$h(t)=t^{nm}+h_{m-1}t^{(m-1)n}+h_{m-2}t^{(m-2)n}+\dots+h_0$, $\nu\not=0$.
Suppose $1\leq l\leq m/2$, $n>m$ and $lm>2$.
Let $S=S_{i_0,n,m,l}(\mathbb{F}_{q^n},\nu,\rho,f)$ with $\nu\neq 0$ and $\mathcal{C}$ be the image of $S$ in  ${\rm End}_{E_f}(R/Rf)$, so that the corresponding rank metric code  $\mathcal{C}$ lies in $M_n(\mathbb{F}_{q^m})$.
\\ (i) We have
$$F\subset\lbrace g\in I_l(\mathcal{C}) \,|\, \text{deg}(g)\leq lm\rbrace= \lbrace g_{i_0}\in K \,|\, \rho(g_{i_0} a) = g_0\rho(a)  \text{ for all } a\in K\rbrace.$$
Moreover, suppose that $i_0=jn \leq lm-1$
and  that $h_j\not=0$. Let $g \in  I_l(\mathcal{C})$.
\\
If $g_{i_0-1}=0$ then $\text{deg}(g)\leq lm$.
\\ If $g_{i_0-1}\not=0$, $h_{lm/n} =0$  and  $g_{lm-1}\not=0 $
then $\text{deg}(g)\leq lm$.  \\
If $\frac{lm}{n}$ is  an integer, $h_{lm/n}\not =0$ and $ g_{i_0-1}\not =h_j g_{lm-1}h_{lm/n}^{-1}$,
then $g \in  I_l(\mathcal{C})$.
\\ (ii) We have
 $$\lbrace g\in I_r(\mathcal{C}) \,|\, \text{deg}(g)\leq lm\rbrace= \lbrace g_0\in \mathbb{F}_{q^n} \,|\, \rho(\sigma^{i_0}(g_0) c_0) = \sigma^{lm}(g_0)\rho(c_0) \text{ for all } c_0\in \mathbb{F}_{q^n}\rbrace.$$
 Moreover, suppose that $i_0=jn \leq lm-1$ and  that $h_j\not=0$. Let $g \in  I_l(\mathcal{C})$.
\\ If $g_{i_0-1}=0$ then $\text{deg}(g)\leq lm$. \\
 If $g_{i_0-1}\not=0$, $h_{lm/n} =0$  and  $g_{lm-1}\not=0 $
then $\text{deg}(g)\leq lm$.
\\ (iii) ${\rm Cent}(\mathcal{C})\cong E_{\hat{h}}=F[x]/ (\hat{h}(x) )$. 
 \\ (iv) 
 $Z(\mathcal{C})$ contains a set of constant skew polynomials that is isomorphic to the set $ \lbrace b\in \mathbb{F}_{q} \,|\, \rho(b a) = b\rho(a)  \text{ for all } a\in \mathbb{F}_{q^n}\rbrace.$
\end{theorem}

\begin{proof}
Let $\mathcal{C}=\{L_a\in {\rm End}_{E_f}(R/Rf)\,|\, a\in P\}$ be the image of $S$ in ${\rm End}_{E_f}(R/Rf)\subset {\rm End}_{F}(R/Rf)$. We identify each element in $\mathcal{C}$ with the element $g\in S$ that induces it.
 \\ (i)
Analogously to the proof of \cite[Theorem 9]{Sheekey},
$$F\subset \lbrace g\in I_l(\mathcal{C}) \,|\, \text{deg}(g)\leq lm\rbrace= \lbrace g_{i_0}\in \mathbb{F}_{q^n} \,|\, \rho(g_{i_0} a) = g_0\rho(a)  \text{ for all } a\in \mathbb{F}_{q^n}\rbrace\subset \mathbb{F}_{q^n}.$$
 To check there are no elements $g\in I_l(\mathcal{C})$ of degree higher than $lm$, we follow the proof of \cite[Theorem 9]{Sheekey}.
 
 Let $g\in I_l(\mathcal{C})$ then $gct^s  \text{ mod } h(t) 
 \in \mathcal{C}$ for all $c\in \lbrace g_0\in \mathbb{F}_{q^n} \,|\, \rho(g_{i_0} a) = g_{i_0}\rho(a)  \text{ for all } a\in \mathbb{F}_{q^n}\rbrace$ and $s\in \{1,\dots , {ml-1}\}$. This is a non-empty set as $lm>1.$

 Consider $s=1$; then 
 $$gt \text{ mod } h(t)=\left(\sum_{i=0}^{mn-1}g_{i-1}t^i\right)-g_{mn-1}\left(\sum_{j=0}^{m-1}h_jt^{nj}\right).$$
As $g\in I_l(\mathcal{C})$, this implies $gt \text{ mod } h\in \mathcal{C}$, so for all $i\in\lbrace lm+1,\dots, mn-1\rbrace$, we have
\begin{equation}\label{eqn:nucleus g_i}
g_{i-1}=\begin{cases}
0& \text{for } i\not\equiv 0\text{ mod }n\\
g_{mn-1}h_{i/n}& \text{for } i\equiv 0\text{ mod }n
\end{cases}
\end{equation}
where $h_{i/n}=0$ if $i/n$ is not an integer.

To try to show that $g_{mn-1}=0$ and thus $\text{deg}(g)\leq lm-1$ we argue verbatim as in  \cite[Theorem 9]{Sheekey} for $lm>2$, and carefully adjusting all degree arguments for an arbitrary fixed $i_0$.

This shows that the argument heavily depends on the choice of $i_0$ and if $i_0 \equiv \,0\, {\rm mod}\, n$ or not. It also depends on if $h_{jn}$ is zero or not when $i_0=jn$.
\\\\
\\ {\bf First case}
For instance choose $i_0=jn\leq lm-1$ and suppose that $h_j\not=0$. Then the coefficient of $t^{i_0}$ in $gt \text{ mod } h(t)$
is given by
$g_{i_0-1}-g_{nm-1}h_j$ which forces that $\rho(g_{i_0-1}-g_{nm-1}h_j)=0$ as in \cite{Sheekey}, since $g_{lm-1}-g_{nm-1}h_{lm/n}=0 $. Thus $g_{i_0-1}=g_{nm-1}h_j$,
and so ${\rm deg}(g)\leq ml-1$ for all $g\in  I_l(\mathcal{C})$ with $g_{i_0-1}=0$. So for all
$g \in  I_l(\mathcal{C})$ with $g_{i_0-1}=0$ it follows that $g\in \{ g\in I_l(\mathcal{C}) \,|\, \text{deg}(g)\leq lm\rbrace$.

 Next assume that $g_{i_0-1}\not=0$. In that case  $g_{nm-1}=h_j^{-1} g_{i_0-1}$ is non-zero. We know that $g_{lm-1}-g_{nm-1}h_{lm/n}=0 $ so substituting yields $g_{lm-1}=h_j^{-1} g_{i_0-1}h_{lm/n} $.
\\
If $\frac{lm}{n}$ is not an integer then $h_{lm/n} =0$ and we obtain $g_{lm-1}=0 $.
\\
If $\frac{lm}{n}$ is  an integer and $h_{lm/n} =0$ then we also obtain $g_{lm-1}=0 $.
\\
If $\frac{lm}{n}$ is  an integer and $h_{lm/n}\not =0$ then we  obtain
$ g_{i_0-1}=h_j g_{lm-1}h_{lm/n}^{-1}$.
\\
  So for all
$g \in  I_l(\mathcal{C})$ with $g_{i_0-1}\not=0$ and additionally one of the following holds:
\\ (1) $\frac{lm}{n}$ is not an integer and  $g_{lm-1}\not=0 $;
\\ (2) $\frac{lm}{n}$ is  an integer, $h_{lm/n} =0$  and  $g_{lm-1}\not=0 $;
\\ (3)  $\frac{lm}{n}$ is  an integer, $h_{lm/n}\not =0$ and $ g_{i_0-1}\not =h_j g_{lm-1}h_{lm/n}^{-1}$,
\\
 it follows that $g\in \{ g\in I_l(\mathcal{C}) \,|\, \text{deg}(g)\leq lm\rbrace$.
We still do not know what happens when $g_{i_0-1}\not=0$,  $h_{lm/n}\not =0$ and $ g_{i_0-1} =h_j g_{lm-1}h_{lm/n}^{-1}$, however.
\\
\\ {\bf Second case} Let $g\in I_l(\mathcal{C})$.
Choose $i_0=jn\leq lm-1$ and suppose that $h_j=0$. Then the coefficient of $t^{i_0}$ in $gt \text{ mod } h(t)$
is given by
$g_{i_0-1}$ which forces  $\rho(g_{i_0-1})=0$, since $g_{lm-1}-g_{nm-1}h_{lm/n}=0 $. Thus $g_{i_0-1}=0$.
So if  $g_{i_0-1}\not=0$ then we get a contradiction and we thus conclude that:
\\ if $h_j=0$ then $g_{i_0-1}=0$. We do not know anything about the degree of $g$ though.
 \\ (ii)
For $I_r(\mathcal{C})$  we obtain that $\lbrace g\in I_r(\mathcal{C}) \,|\, \text{deg}(g)\leq lm\rbrace \subset \mathbb{F}_{q^n}$ following Sheekey's proof. Since $(c_{i_0}+\nu\rho(c_{i_0})t^{lm}) g_0 \in \mathcal{C}$ if and only if $\nu\rho(c_{i_0}) \sigma^{lm}(g_0)=\nu\rho(c_{i_0}g_0)$, which is equivalent to $\rho(c_{i_0}g_0)= \sigma^{lm}(g_0) \rho(c_{i_0}) $, we have
$$\lbrace g\in I_r(\mathcal{C}) \,|\, \text{deg}(g)\leq lm\rbrace= \lbrace g_0\in \mathbb{F}_{q^n} \,|\, \rho(\sigma^{i_0}(g_0) c_0) = \sigma^{lm}(g_0)\rho(c_0)
\text{ for all } c_0\in \mathbb{F}_{q^n}\rbrace.$$

 Let $g\in I_r(\mathcal{C})$ then $t^s g \text{ mod } h(t)
 \in \mathcal{C}$ for all $s\in \{1,\dots , {ml-1}\}$. This is a non-empty set as $lm>1.$  Consider $s=1$; then
 $$tg \text{ mod } h(t)=\left(\sum_{i=0}^{mn-1}\sigma^i(g_{i-1})t^i\right)-\left(\sum_{j=0}^{m-1}h_jt^{nj}\right) g_{mn-1}.$$
As $g\in I_r(\mathcal{C})$, this implies $t g\text{ mod } h\in \mathcal{C}$, so for all $i\in\lbrace lm+1,\dots, mn-1\rbrace$, we have
\begin{equation}\label{eqn:nucleus g_i}
\sigma^i(g_{i-1})=\begin{cases}
0& \text{for } i\not\equiv 0\text{ mod }n\\
 \sigma^{i/n}(g_{mn-1}) h_{i/n} & \text{for } i\equiv 0\text{ mod }n
\end{cases}
\end{equation}
where $h_{i/n}=0$ if $i/n$ is not an integer.

To try to show that $g_{mn-1}=0$ and thus $\text{deg}(g)\leq lm-1$ we argue verbatim as in $(i)$:
\\
\\ {\bf First case}
Choose $i_0=jn\leq lm-1$ and suppose that $h_j\not=0$. Then the coefficient of $t^{i_0}$ in $gt \text{ mod } h(t)$
is given by
$\sigma^{i_0}(g_{i_0-1})-\sigma^j(g_{nm-1})h_j$ which forces that $\rho(\sigma^{i_0}(g_{i_0-1})-\sigma^j(g_{nm-1})h_j)=0$, since $\sigma^{lm+1}(g_{lm-1})-\sigma^{j+1}(g_{nm-1}h_{lm/n})=0 $.

 Thus $\sigma^{i_0}(g_{i_0-1})=\sigma^{j+1}(g_{nm-1})h_j$,
and so ${\rm deg}(g)\leq ml-1$ for all $g\in  I_r(\mathcal{C})$ with $g_{i_0-1}=0$. So for all
$g \in  I_r(\mathcal{C})$ with $g_{i_0-1}=0$ it follows that $g\in \{ g\in I_r(\mathcal{C}) \,|\, \text{deg}(g)\leq lm\rbrace$.
\\ $(iii)$
The result that ${\rm Cent}(\mathcal{C})\cong E_{\hat{h}}$ does not rely on $\rho$ being an automorphism nor on the choice of $i_0$,
and  follows verbatim  as in the proof of \cite[Theorem 9]{Sheekey}, also $(iv)$ does not rely on $\rho$ being an automorphism nor on the choice of $i_0$, only on $I_r(\mathcal{C})$ which is not fully computed above.
\end{proof}

\begin{corollary}\label{cor:nuclei K2}
Let $f\in R=\mathbb{F}_{q^n}[t;\sigma]$ be monic irreducible, $\nu\neq 0$, $l=1$ and 
$h(t)=t^{nm}+h_{m-1}t^{(m-1)n}+h_{m-2}t^{(m-2)n}+\dots+h_0$.  Suppose that $n>1$, $n>m$ and $m>2$.
 Then the following holds for the algebra $S=S(\mathbb{F}_{q^n},i_0,\nu,\rho,f)$:
\\ (i) $F\subset  \{g_{i_0}\in \mathbb{F}_{q^n} \,|\, \rho(g_{i_0} a) = g_0\rho(a)  \text{ for all } a\in K\} =\{g\in {\rm Nuc}_l(S)\,|\, \text{deg}(g)\leq m \}\subset {\rm Nuc}_l(S)$.
\\
Moreover, assume that $i_0=jn \leq m-1$ and  $h_j\not=0$. Let $g \in {\rm Nuc}_l(S)$.
If $g_{i_0-1}=0$ then $\text{deg}(g)\leq m$.
 If $g_{i_0-1}\not=0$, $h_{m/n} =0$  and  $g_{m-1}\not=0 $
then $\text{deg}(g)\leq m$.  If $\frac{m}{n}$ is  an integer, $h_{m/n}\not =0$ and $ g_{i_0-1}\not =h_j g_{m-1}h_{m/n}^{-1}$,
then $g \in  {\rm Nuc}_l(S)$.
\\ (ii)
$  \lbrace g_0\in K \,|\, \rho(\sigma^{i_0}(g_0) c_0) = \sigma^{m}(g_0)\rho(c_0) \text{ for all } c_0\in K\rbrace =
\{ g\in {\rm Nuc}_m(S) \,|\, \text{deg}(g)\leq m\}\subset {\rm Nuc}_m(S).$
Moreover, suppose that $i_0=jn \leq lm-1$ and  that $h_j\not=0$. Let $g \in {\rm Nuc}_m(S)$.
If $g_{i_0-1}=0$ then $\text{deg}(g)\leq lm$.  If $g_{i_0-1}\not=0$, $h_{lm/n} =0$  and  $g_{lm-1}\not=0 $
then $\text{deg}(g)\leq lm$.
\\ (iii) The center $ C(S)$ contains the subset $\lbrace b\in F' \,|\, \rho(b a) = b\rho(a) \text{ for all } a\in K\rbrace,$
 so $S$ is an algebra over $F'$
\\ (iv) ${\rm Nuc}_r(S)=\mathbb{F}_{q^m}$.
\end{corollary}

\begin{corollary}\label{cor:nuclei K2} 
Let $f\in R=\mathbb{F}_{q^n}[t;\sigma]$ be monic irreducible, $\nu\neq 0$, , $l=1$, 
$h(t)=t^{nm}+h_{m-1}t^{(m-1)n}+h_{m-2}t^{(m-2)n}+\dots+h_0$, and $\rho\in {\rm Aut}(\mathbb{F}_{q^n})$.
Suppose that $n>1$, $n>m$ and $m>2$.
 Then the following holds for the non-unital ring  $S=\mathbb{S}_f(K,i_0,\nu,\rho)$:
\\ (i) ${\rm Fix}(\rho)=\{g\in {\rm Nuc}_l(S)\,|\, \text{deg}(g)\leq m\}\subset  {\rm Nuc}_l(S)$.
\\
Moreover, suppose that $i_0=jn \leq m-1$ and  that $h_j\not=0$. Let $g \in {\rm Nuc}_l(S)$.\\
If $g_{i_0-1}=0$ then $\text{deg}(g)\leq m$.\\
 If $g_{i_0-1}\not=0$, $h_{m/n} =0$  and  $g_{m-1}\not=0 $
then $\text{deg}(g)\leq m$.
\\ If $\frac{m}{n}$ is  an integer, $h_{m/n}\not =0$ and $ g_{i_0-1}\not =h_j g_{m-1}h_{m/n}^{-1}$,
then $g \in  {\rm Nuc}_l(S)$.
\\ (ii) ${\rm Fix}(\rho^{-1}  \circ\sigma^{m-i_0})={\rm Fix}(\rho^{-1}  \circ\sigma^{m-i_0}) 
= \{ g\in {\rm Nuc}_m(S) \,|\, \text{deg}(g)\leq m\}\subset  {\rm Nuc}_m(S).$
Moreover, suppose that $i_0=jn \leq lm-1$ and  that $h_j\not=0$. Let $g \in {\rm Nuc}_m(S)$.\\
If $g_{i_0-1}=0$ then $\text{deg}(g)\leq lm$.  If $g_{i_0-1}\not=0$, $h_{lm/n} =0$  and  $g_{lm-1}\not=0 $
then $\text{deg}(g)\leq lm$.
\\ (iii) 
${\rm Nuc}_r(S)=\mathbb{F}_{q^m}$.
\end{corollary}

Note that in $(ii)$ we used that $\rho(\sigma^{i_0}(g_0)) = \sigma^{m}(g_0)$ for all $g_0\in \mathbb{F}_{q^n}$ is equivalent to $\sigma^{-m}(\rho(\sigma^{i_0}(g_0)) )=
g_0$ for all $g_0\in \mathbb{F}_{q^n}$, i.e. to $g_0\in {\rm Fix}(\sigma^{-m}\circ \rho \circ\sigma^{i_0}).$

Alternatively, to stay in line with \cite{Sheekey},
$\rho(\sigma^{i_0}(g_0)) = \sigma^{m}(g_0)$ for all $g_0\in \mathbb{F}_{q^n}$ is equivalent to $\sigma^{i_0}(g_0) = \rho^{-1}\sigma^{m}(g_0)$ for all $g_0\in \mathbb{F}_{q^n}$, and if $\sigma$ and $\rho$ commute, e.g. when $F$ is a finite field like here, then this is the same as $g_0\in {\rm Fix}(\rho^{-1}  \circ\sigma^{m-i_0}).$  
\\\\
{\bf Conclusion:}
If all $g\in  {\rm Nuc}_m(S)$ have degree less or equal to $m$ then $ {\rm Nuc}_m(S)$  differs from the one in the pre-semifield construction of Sheekey (where $i_0=0$), whenever ${\rm Fix}(\rho^{-1}  \circ\sigma^{m-i_0})\not = {\rm Fix}(\rho^{-1}  \circ\sigma^{m}).$

If there are $g\in  {\rm Nuc}_m(S)$ of degree greater than $m$ then $ {\rm Nuc}_m(S)$  differs from the one of Sheekey's generalised twisted cyclic algebras as well, while the algebras have the same centraliser and right nucleus.

\section{Some examples of pre-semifields and MRD codes} \label{sec:finite fields}

Let $F=\mathbb{F}_q$, $K=\mathbb{F}_{q^n}$, and let $f\in \mathbb{F}_{q^n}[t;\sigma]$ be a monic and irreducible polynomial of degree $m$.

 \subsection{A worked example where $q=3$, $n=m=2$.} \label{sec:ex}

  For $q$ odd, every semifield of order $q^4$ which has center $\mathbb{F}_{q}$ and right nucleus $\mathbb{F}_{q^2}$, is isotopic to a pre-semifield $\mathbb{S}_f(\mathbb{F}_{q},i_0=0,\nu,\rho)$
  for a suitable irreducible $f$ of degree 2, some $\nu$ and some automorphism $\rho\in {\rm Gal}(\mathbb{F}_{q^2}/\mathbb{F}_{q})$, or isotopic to a pre-semifield $\mathbb{S}_f(\mathbb{F}_{q},i_0=0,\nu,\rho)$ for some linear $f$, some $\nu$ and some automorphism $\rho\in {\rm Aut}(\mathbb{F}_{q^4})$ \cite[Theorem 8]{Sheekey}.

   Hence for quadratic irreducible monic $f\in \mathbb{F}_{q^2}[t;\sigma]$,  choosing $i_0=1$  and constructing the pre-semifields $\mathbb{S}_f(\mathbb{F}_{q^n},i_0,\nu,\rho)$ cannot yield any new semifields  of order $q^4$ with center $\mathbb{F}_{q}$ and right nucleus $\mathbb{F}_{q^2}$, if we look at semifields that are isotopic to these pre-semifields. However, the methods we introduced are easiest to show for skew polynomials $f$ of degree two, and we will do so in the following. The computational effort increases for higher degrees.

 Let $q=3$, $n=m=2$ and choose $\alpha=1+i$ which is a primitive element of $\mathbb{F}_9=\mathbb{F}_3(i)$. Let ${\rm Gal}(\mathbb{F}_9/\mathbb{F}_3)=\langle \sigma \rangle$, $\sigma(x)=x^3$. Then the monic irreducible polynomials in $\mathbb{F}_9[t;\sigma]$ can be put into three classes of mutually similar polynomials.
\newpage
Class A:

        \ignore{$t^2 + \alpha^6t + 2,$

        $t^2 + 2t + 2,$
$t^2 + \alpha^2t+ 2,$
$t^2 + t + 2,$
$t^2 + \alpha^2,$
$t^2 + \alpha^7t + 1,$
$t^2 + \alpha^5t + 1,$
$t^2 + \alpha^6,$
$t^2 + \alpha^3t + 1,$
$ t^2 + \alpha t + 1$.}
$t^2 + \alpha^6t + 2,$
$t^2 + 2t + 2,$
$t^2 + \alpha^2t+ 2,$ $t^2 + t + 2,$
$t^2 + \alpha^2,$
$t^2 + \alpha^7t + 1,$
$t^2 + \alpha^5t + 1,$
$t^2 + \alpha^6,$
$t^2 + \alpha^3t + 1,$
$ t^2 + \alpha t + 1$.
\\\\
     Class B:

       $ t^2 + \alpha^7t+ \alpha,$
$t^2 + \alpha^7,$
$ t^2 + \alpha^3 t + \alpha,$
$t^2 + \alpha^5t + \alpha^3,$
 $t^2 + \alpha t + \alpha^3,$
$ t^2 + \alpha^5,$
 $t^2 + \alpha^5t + \alpha,$
 $ t^2 + \alpha t + \alpha,$
$  t^2 + \alpha^7t + \alpha^3,$
$  t^2 + \alpha^3t + \alpha^3$
\\\\
Class C:

       $ t^2 + \alpha^3,$
$ t^2 + \alpha^6t + \alpha^5,$
$ t^2 + 2 t + \alpha^5,$
$ t^2 + \alpha^2t + \alpha^5,$
$ t^2 + x + \alpha^5,$
$ t^2 + \alpha,$
$ t^2 + \alpha^6t + \alpha^7,$
 $t^2 + 2 t + \alpha^7,$
 $t^2 + \alpha^2t + \alpha^7,$
$t^2 + t + \alpha^7$.

We now construct some  pre-semifield $\mathbb{S}_f(K,i_0,\nu,\rho)$ of order $81$  with center $\mathbb{F}_3$.

\begin{lemma}
Let $f$ be in Class A and let $\rho$ be either the identity or the nontrivial $\sigma$ generating the Galois group of $\mathbb{F}_9/\mathbb{F}_3$. Then $\mathbb{S}_f(K,i_0=0,\nu,\rho)$ is a pre-semifield
if and only if
$\nu \in \{ \alpha, \alpha^3,\alpha^5, \alpha^7  \}$.
\end{lemma}

\begin{proof}
 We have $C_{f,{0}}=\{ 1,\alpha^2, \alpha^4, \alpha^6  \}$.

Let $\rho=id$ then $g_{0}\not =\nu^{-1}$ if and only if $\nu^{-1}\not\in \{ 1, \alpha^2,\alpha^4, \alpha^6  \}$ if and only if
 $\nu^{-1}\in \{ \alpha, \alpha^3,\alpha^5, \alpha^7 \}$ if and only if $\nu\in \{ \alpha, \alpha^3,\alpha^5, \alpha^7 \}$. Therefore by Theorem \ref{thm:division2}, for any $f$ in Class A,
$\mathbb{S}_f(K,i_0=0,\nu,id)$  is a  pre-semifeld over $\mathbb{F}_3$ if and only if $\nu\in \{ \alpha, \alpha^3,\alpha^5, \alpha^7 \}$.

Moreover, the norm of the constant factor of each $f$ in Class A is one. Hence for each  $f$ in Class A and all $\nu$ such that $N_{K/F}(\nu)=2$, that means for all $\nu \in \{ \alpha, \alpha^3,\alpha^5, \alpha^7 $, $\mathbb{S}_f(K,i_0=0,\nu,id)$ is a pre-semifield over $\mathbb{F}_3$  by Sheekey
\cite[Theorem 7]{Sheekey}.  In this case, Corollary \ref{cor:Sh2}, Corollary \ref{cor:Sh} and \cite[Theorem 7]{Sheekey}  yield the same pre-semifields.

Let $\rho=\sigma$. By Theorem \ref{thm:division2},
$\mathbb{S}_f(K,i_0=0,\nu,\sigma)$  is a pre-semifield  if and only if  for all
 monic polynomials  $g(t)=g_0+g_1t+t^2\in R$ that are similar to $f$, either $\sigma(g_{0})\not =\nu^{-1}$ and $g_{0}\not=0$, or $g_{0}=0$.
 For nonzero $g_0$, here $\sigma(g_{0})\in\{ 1,\alpha^2, \alpha^4, \alpha^6  \}$, so again $\mathbb{S}_f(K,i_0=0,\nu,\sigma)$  is a pre-semifeld over $\mathbb{F}_3$ if and only if $\nu\in \{ \alpha, \alpha^3,\alpha^5, \alpha^7 \}$.

\end{proof}

By Theorem \ref{thm:division2},
$\mathbb{S}_f(K,i_0,\nu,\rho)$  is a division algebra  if and only if  for all
 monic polynomials  $g(t)=g_0+g_1t+t^2\in R$ that are similar to $f$, either $\nu\rho(g_{i_0})\not =1$ and $g_{i_0}\not=0$, or $g_{i_0}=0$.

 Let $i_0=1$, then  $C_{f,{1}}=\mathbb{F}_9^\times$ is biggest possible. Since $\rho$ is bijective, also $\{\rho(g_{1})\,|\, g_1\in \mathbb{F}_9^\times\}=\mathbb{F}_9^\times$ and so for every $\nu$ there is some monic polynomial $g(t)=g_0+g_1t+t^2\in R$ similar to $f$
 such that $g_{1}\not=0$ and $\rho(g_{1}) =\nu^{-1}$, i.e. $\nu\rho(g_{1}) =1$.
Therefore
$\mathbb{S}_f(K,i_0=1,\nu,\rho)$ is not a pre-semifield for any $f$ in Class A.

\begin{lemma}
Let $f$ be in Class B and let $\rho$ be either the identity or the nontrivial $\sigma$ generating the Galois group of $\mathbb{F}_9/\mathbb{F}_3$. Then for each $f$ in Class C and  for all $\nu\in \{ 1,2,i,2i \}$,
$\mathbb{S}_f(K,i_0=1,\nu,\rho)$ and $\mathbb{S}_f(K,i_0=0,\nu,\rho)$ are pre-semifields.
\end{lemma}

\begin{proof}
  Let $i_0=1$.   For any $f$ in Class B, we have $C_{f,1}=\{ \alpha, \alpha^3, \alpha^5,\alpha^7\}.$
  By Theorem \ref{thm:division2},
$\mathbb{S}_f(K,i_0=1,\nu,\rho)$  is a pre-semifield  if and only if  for all
 monic polynomials  $g(t)=g_0+g_1t+t^2\in R$ that are similar to $f$, either $\nu\rho(g_{1})\not =1$ and $g_{1}\not=0$, or $g_{1}=0$.

 Let $\rho=id$ then $g_{1}\not =\nu^{-1}$ if and only if $\nu^{-1}\not\in \{ \alpha, \alpha^3, \alpha^5,\alpha^7\}$ which means $\nu\in \{ 1,\alpha^2, \alpha^4, \alpha^6\}=\{ 1,2,i,2i\}$, hence the algebra  $\mathbb{S}_f(K,i_0=1,\nu,id)$ is a pre-semifield over $\mathbb{F}_3$ for all
 $\nu\in \{ 1,\alpha^2, \alpha^4, \alpha^6\}=\{ 1,2,i,2i\}$.

 Now let $\rho=\sigma$, then for nonzero $g_0$,  $\sigma(g_{0})\in\{ \alpha, \alpha^3, \alpha^5, \alpha^7  \}$, so again $\mathbb{S}_f(K,i_0=1,\nu,\sigma)$  is a pre-semifeld over $\mathbb{F}_3$ if and only if $\nu\in \{ 1,\alpha^2, \alpha^4,\alpha^6\}=\{1,2, i,2i\}$.

Let $i_0=0$.  We have $C_{f,0}=\{ \alpha, \alpha^3, \alpha^5,\alpha^7\}.$
   The analogous calculation as for $i_0=1$ yields that
   $\mathbb{S}_f(K,i_0=0,\nu,\sigma)$ is a pre-semifield if and only if $\nu\in \{ 1, \alpha^4,\alpha^6,\alpha^2 \}=\{ 1,2,i,2i \}$.

     The norm of the constant factor of each $f$ in Class B is one. If we choose $i_0=0$ and $\rho=id$, for each such $f$ and all $\nu$ such that $N(\nu)=1$, i.e. for all $\nu\in \{ 1,2,i,2i \}$, the algebra $\mathbb{S}_f(K,i_0=0,\nu,id)$ is a pre-semifield over $\mathbb{F}_3$ \cite[Theorem 7]{Sheekey}. We obtain the exact same result applying our Corollary \ref{cor:Sh}. We can also argue exactly as when $i_0=1$ to obtain that
  for all $\nu\in \{ 1,\alpha^2, \alpha^4, \alpha^6\}=\{ 1,2,i,2i\}$, the algebra $\mathbb{S}_f(K,i_0=1,\nu,id)$ is a pre-semifield over $\mathbb{F}_3$.
\end{proof}

\begin{lemma}
Let $f$ be in Class C and let $\rho$ be either the identity or the nontrivial $\sigma$ generating the Galois group of $\mathbb{F}_9/\mathbb{F}_3$.
Then $\mathbb{S}_f(K,i_0=1,\nu,\rho)$ is a pre-semifield for each $f$ in Class C and  for all $\nu\in \{ \alpha, \alpha^3, \alpha^5,\alpha^7\}$, and $\mathbb{S}_f(K,i_0=0,\nu,\rho)$ is a pre-semifield for each $f$ in Class C and  for all $\nu\in \{ 1,2,i,2i \}$.
\end{lemma}

\begin{proof}
      Let $i_0=1$. For any $f$ in Class C, we have   $C_{f,1} =\{1,2,   \alpha^2,  \alpha^6\}=\{1,2, 2i, i\}$.
      For all $g(t)=g_0+g_1t+\nu\rho(g_{1})t^2\in P$ with $g_{1}\not=0$ which are similar to $f$, we need $ g_{1}\nu^{-1}\rho(g_{1})^{-1} \in
      \{ \alpha, \alpha^3, \alpha^5,\alpha^7\}$ in order for $\mathbb{S}_f(K,i_0,\nu,\rho)$ to be a pre-semifield.

Let $\rho=id$, then $\mathbb{S}_f(K,i_0=1,\nu,id)$ is a pre-semifield for all $\nu\in \{ \alpha, \alpha^3, \alpha^5,\alpha^7\}$.

Let $\rho=\sigma$, then we need $g_{1}\nu^{-1}\sigma(g_{1})^{-1} \in \{ \alpha, \alpha^3, \alpha^5,\alpha^7\}$ in order for $\mathbb{S}_f(K,i_0=1,\nu,\sigma)$ to be a pre-semifield. We compute
$$g_{1}\sigma(g_{1})^{-1} \in \{1,\alpha^4\},\quad g_{1}\alpha\sigma(g_{1})^{-1}\in  \{\alpha, \alpha^5\},
\quad g_{1}\alpha^3\sigma(g_{1})^{-1}\in  \{\alpha^3,\alpha^7\},
 $$
$$g_{1}\alpha^5\sigma(g_{1})^{-1}\in  \{\alpha^5,\alpha\}, \quad  g_{1}\alpha^7\sigma(g_{1})^{-1}\in \{\alpha^7,\alpha^3\}.$$
Hence $\mathbb{S}_f(K,i_0=1,\nu,\sigma)$ is a pre-semifield for all $\nu\in \{ \alpha, \alpha^3, \alpha^5,\alpha^7\}$.

Let $i_0=0$. For any $f$ in Class C, we have   $C_{f,0} =\{ \alpha, \alpha^3, \alpha^5,\alpha^7\}$ which is the same as $C_{f,0}$ for class B. The analogous calculations yield   that
   $\mathbb{S}_f(K,i_0=0,\nu,\sigma)$ is a pre-semifield over $\mathbb{F}_3$  for each such $f$ in Class C and for all $\nu\in \{ 1, \alpha^4,\alpha^6,\alpha^2 \}=\{ 1,2,i,2i \}$ (Corollary \ref{cor:Sh}), and that
   $\mathbb{S}_f(K,i_0=0,\nu,id)$ is a pre-semifield over $\mathbb{F}_3$  for each such $f$ in Class C and for all $\nu\in \{ 1,\alpha^2, \alpha^4, \alpha^6\}=\{ 1,2,i,2i\}$.
\end{proof}

 By \cite[Theorem 7]{Sheekey}, $\mathbb{S}_f(K,i_0=0,\nu,id)$ is a pre-semifield over $\mathbb{F}_3$ for each $f$ in Class C and all $\nu$ such that $N_{K/F}(\nu)=1$, i.e. for all $\nu\in \{ 1,2,i,2i \}$, so the above results coincide with this observation.

\subsection{A worked example where $m=n$ prime and $f(t)=t^n-a_0\in \mathbb{F}_{q^n}[t;\sigma]$.}
Let $K=\mathbb{F}_{q^n}=\mathbb{F}_{q}(a_0)$ and $n$ be a prime. Let $f(t)=t^n-a_0\in \mathbb{F}_{q^n}[t;\sigma]$ be irreducible.
E.g., if $F$ contains a primitive $n$th root of unity, then $f$ is irreducible if and only if $a_0 \not=\sigma^{n-1}(z)\cdots\sigma(z)z$ for all $z\in K$.

Let $\rho$ be any $F'$-bijective linear map such that $F'\subset \mathbb{F}_{q}$. Fix $\nu \in \mathbb{F}_{q^n}^\times$, and choose  $i_0\in \{0,1,\dots, n-1\}$.
We compute the linear code associated to the $n^2[F:F']$-dimensional  $F'$-algebra  $S=\mathbb{S}_f(\mathbb{F}_{q^n},i_0,\nu,\rho)$:

Analogously as in \cite[Section 8]{PT2} we obtain the matrix spread set
$$\mathcal{C}_{i_0, n,n,1}=\lbrace M_a\mid a_k\in \mathbb{F}_{q^n} \mbox{ for } k=0,1,\dots,n-1\rbrace\subset M_n(\mathbb{F}_{q^n})$$
of  $S$, where
$M_a=(m_{i,j})_{i,j}$ is defined as follows:
\begin{equation*}
 m_{i,j}=
\begin{cases}
\sigma^{n+1-i}(a_0)+\sigma^{n+1-i}(\nu\rho(a_{i_0}))a_0&\mbox{for } i=j,\\
\sigma^{n+1-i}(a_{i-j}) &\mbox{for } i>j,\\
\sigma^{n+1-i}(a_{n+i-j})a_0&\mbox{for } i<j.
\end{cases}
\end{equation*}
The code $\mathcal{C}_{i_0, n,n,1}$ is an $F'$-linear code and an MRD-code, if and only if the set $ P_{i_0} $ does not contain any polynomial similar to $f$.

 For $n>2$  we know that ${\rm Nuc}_r(S)=\mathbb{F}_{q^n}.$

 Moreover, we got the following partial results:
  $$\{g_{i_0}\in \mathbb{F}_{q^n} \,|\, \rho(g_{i_0} a) = g_0\rho(a)  \text{ for all } a\in \mathbb{F}_{q^n}\} =\{g\in {\rm Nuc}_l(S)\,|\, \text{deg}(g)\leq m \}\subset {\rm Nuc}_l(S).$$
Assume that $i_0=jn \leq m-1$ and  $h_j\not=0$. Let $g \in {\rm Nuc}_l(S)$.
If $g_{i_0-1}=0$ then $\text{deg}(g)\leq m$.
 If $g_{i_0-1}\not=0$, $h_{m/n} =0$  and  $g_{m-1}\not=0 $
then $\text{deg}(g)\leq m$.  If $\frac{m}{n}$ is  an integer, $h_{m/n}\not =0$ and $ g_{i_0-1}\not =h_j g_{m-1}h_{m/n}^{-1}$,
then $g \in  {\rm Nuc}_l(S)$. Furthermore,
$$  \lbrace g_0\in K \,|\, \rho(\sigma^{i_0}(g_0) c_0) = \sigma^{m}(g_0)\rho(c_0) \text{ for all } c_0\in K\rbrace =
\{ g\in {\rm Nuc}_m(S) \,|\, \text{deg}(g)\leq m\}\subset {\rm Nuc}_m(S).$$
In particular, suppose that $i_0=jn \leq lm-1$ and  that $h_j\not=0$. Let $g \in {\rm Nuc}_m(S)$.
If $g_{i_0-1}=0$ then $\text{deg}(g)\leq lm$.  If $g_{i_0-1}\not=0$, $h_{lm/n} =0$  and  $g_{lm-1}\not=0 $
then $\text{deg}(g)\leq lm$.

 \begin{example}
 The skew polynomial
$f(t)=t^2-a_0\in \mathbb{F}_{q^2}[t;\sigma]$ has $h(t)=f(t) (t^2-\sigma(a_0 ))=t^4-T_{\mathbb{F}_{q^2}/\mathbb{F}_{q}}(a_0)t^2+ N_{\mathbb{F}_{q^2}/\mathbb{F}_{q}}(a_0 )$ and is irreducible if and only if $a_0 \not\in N_{\mathbb{F}_{q^2}/\mathbb{F}_{q}}(\mathbb{F}_{q^2}^\times)$.

Let $f$ be monic irreducible of degree two and not two-sided.

Choose $\rho=id$, then $S=\mathbb{S}_f(\mathbb{F}_{q^2},i_0=1,\nu,id,f)$ is a pre-semifield of order $q^4$  with center containing $\mathbb{F}_q$ for any choice of $\nu\in \mathbb{F}_q^\times$. We have $\mathbb{F}_{q^2}\subset {\rm Nuc}_l(S)$, and
$\mathbb{F}_q={\rm Fix}(\sigma)=\{ g\in {\rm Nuc}_m(S) \,|\, \text{deg}(g)\leq 2\}\subset  {\rm Nuc}_m(S).$
Its matrix spread set gives the following MRD code:
$$\mathcal{C}_{1,2,2,1}=\Big\{
\begin{pmatrix}
z_0+\nu z_1 a_0& z_1 a_0\\
\sigma(z_1)& \sigma(z_0)+ \sigma(\nu)\sigma(z_1) a_0
\end{pmatrix}
\mid z_0,z_1\in \mathbb{F}_{q^2}
\Big\}.$$

Choose $\rho=\sigma$. Then $S=\mathbb{S}_f(\mathbb{F}_{q^2}, i_0=1,\nu,\sigma)$ is a pre-semifield of order $q^4$  with  center $\mathbb{F}_q$ for any choice of $\nu\in \mathbb{F}_q^\times$, and
$\mathbb{F}_{q^2}={\rm Fix}(id)\subset  {\rm Nuc}_m(S)$. That means $\mathbb{F}_{q^2}\subset  {\rm Nuc}_m(S)$, so that $\mathbb{F}_{q^2}=  {\rm Nuc}_m(S)$ when $\rho=id$, and
$\mathbb{F}_{q}\subset  {\rm Nuc}_m(S)$ when $\rho=\sigma$.
Its matrix spread set gives the following MRD code:
 $$\mathcal{C}_{1,2,2,1}=\Big\{
\begin{pmatrix}
z_0+\nu\sigma(z_1)a_0& z_1 a_0\\
\sigma(z_1)& \sigma(z_0)+ \sigma(\nu)z_1 a_0
\end{pmatrix}
\mid z_0,z_1\in \mathbb{F}_{q^2}
\Big\}.$$
In particular, choose  $F=\mathbb{F}_3$, $K=\mathbb{F}_9$,  $\alpha=1+i$ a primitive element.
By Section \ref{sec:ex},
either  $f(t)=t^2+(2+i)$,  $f(t)=t^2+(2+2i)$ (both in Class B),  $f(t)=t^2+(1+i)$, or $f(t)=t^2+(1+2i)$ (both in Class C).
 For the right choices of an $\mathbb{F}_3$-linear $\rho$ and $\nu\in \mathbb{F}_{9}^\times$,
$S_{t^2-a_0}(\mathbb{F}_{q^2},i_0=1,\nu,\rho)$  is a four-dimensional division algebra with center $\mathbb{F}_3$ and right nucleus $\mathbb{F}_{9}$.
Its matrix spread set yields the following MRD code:
$$\mathcal{C}_{1,2,2,1}=\Big\{
\begin{pmatrix}
z_0+\nu\rho(z_1)a_0& z_1 a_0\\
\sigma(z_1)& \sigma(z_0)+ \sigma(\nu)\sigma(\rho(z_1)) a_0
\end{pmatrix}
\mid z_0,z_1\in \mathbb{F}_{q^2}
\Big\}.$$
For $\rho\in {\rm Aut}(\mathbb{F}_{9})$ we have
${\rm Fix}(\rho)=\{g\in {\rm Nuc}_l(S)\,|\, \text{deg}(g)\leq 2\}\subset  {\rm Nuc}_l(S)$ and
${\rm Fix}(\rho \circ\sigma^{-1})=\{ g\in {\rm Nuc}_m(S) \,|\, \text{deg}(g)\leq 2\}\subset  {\rm Nuc}_m(S).$ 
 \end{example}

\begin{example}
 For $n=m=3$,  $f(t) = t^3-a_0\in K[t;\sigma]$  has $h(t)=f(t)(t^3-\sigma(a_0 ))(t^3-\sigma^2(a_0 ))$
  and is irreducible if and only if $a_0 \not=\sigma^{2}(z)\sigma(z)z$ for all $z\in K$.

  We have the choice of ${i_0}\in \{ 0,1,2\}$, $i_0=0$ has been done in \cite[Section 4.2]{Sheekey}. The left multiplication in the algebra
  $$\mathbb{S}_{t^3-a_0}(K,i_0=1,\nu,\rho)$$
   gives the following set $\mathcal{C}_{1,3,3,1}$ of matrices:
 $$\Big\{
 \left( \begin{array}{ccc}
z_0+\nu\rho(z_{i_0})a_0&       z_2      a_0                                       &   z_1 a_0\\
\sigma^2(z_1)      &     \sigma^2(z_0)+ \sigma^2(\nu)\sigma^2(\rho(z_{i_0})) a_0  &   \sigma^2(z_2)a_0 \\
\sigma(z_2)        &      \sigma(z_1)                                         & \sigma(z_0)+ \sigma(\nu)\sigma(\rho(z_{i_0})) a_0
\end{array} \right)
\mid z_i\in K
\Big\}.
$$
Let $K=\mathbb{F}_{q^3}$ and $F'=\mathbb{F}_q$.  Then $S= \mathbb{S}_{t^3-a_0}(\mathbb{F}_{q^3},i_0,\nu,\rho)$ is an algebra with $q^9$ elements and center $\mathbb{F}_q$
and if it is a division algebra then  its matrix spread set gives the above MRD code where now $z_i\in \mathbb{F}_{q^3}$.

For any $\rho\in {\rm Aut}(\mathbb{F}_{q^3})$, the associated semifield $S$ has
$|{\rm Nuc}_r (S)|=q^3$. Its middle nucleus contains ${\rm Fix}(\rho^{-1}\circ \sigma^{3-i_0})$
and its left nucleus contains ${\rm Fix}(\rho)$ which is either $\mathbb{F}_{q^3}$ (when $\rho=id$), or $\mathbb{F}_{q}$.

Choose $\rho=id$. Then
 $$\mathcal{C}_{1,3,3,1}= \Big\{
 \left( \begin{array}{ccc}
z_0+\nu z_{i_0}a_0&       z_2      a_0                                       &   z_1 a_0\\
\sigma^2(z_1)      &     \sigma^2(z_0)+ \sigma^2(\nu)\sigma^2(z_{i_0}) a_0  &   \sigma^2(z_2)a_0 \\
\sigma(z_2)        &      \sigma(z_1)                                         & \sigma(z_0)+ \sigma(\nu)\sigma( z_{i_0}) a_0
\end{array} \right)
\mid z_i\in \mathbb{F}_{q^3}
\Big\}
$$
and the associated semifield has left nucleus  $\mathbb{F}_{q^3}$.
Choose $\rho=\sigma$ then
 $$\mathcal{C}_{1,3,3,1}= \Big\{
 \left( \begin{array}{ccc}
z_0+\nu\sigma(z_{i_0})a_0&       z_2      a_0                                       &   z_1 a_0\\
\sigma^2(z_1)      &     \sigma^2(z_0)+ \sigma^2(\nu)z_{i_0}) a_0  &   \sigma^2(z_2)a_0 \\
\sigma(z_2)        &      \sigma(z_1)                                         & \sigma(z_0)+ \sigma(\nu)\sigma^2(z_{i_0})) a_0
\end{array} \right)
\mid z_i\in \mathbb{F}_{q^3}
\Big\}.
$$
Choose $\rho=\sigma^2$ then
$$\mathcal{C}_{1,3,3,1}= \Big\{
 \left( \begin{array}{ccc}
z_0+\nu \sigma^2(z_{i_0})a_0&       z_2      a_0                                       &   z_1 a_0\\
\sigma^2(z_1)      &     \sigma^2(z_0)+ \sigma^2(\nu)\sigma(z_{i_0}) a_0  &   \sigma^2(z_2)a_0 \\
\sigma(z_2)        &      \sigma(z_1)                                         & \sigma(z_0)+ \sigma(\nu)z_{i_0} a_0
\end{array} \right)
\mid z_i\in \mathbb{F}_{q^3}
\Big\}.
$$
\end{example}

\section{Conclusions}

To summarize, our setup is more general than  the one that appears  in the current  literature:
$(i)$
 We use linear  maps and not just automorphisms  in the definitions. $(ii)$ We use different choices for $i_0$. $(iii)$  We do not always assume that the minimal central left multiple $h$ of $f$ (i.e. the bound of $f$)
 has  maximal degree.

If desired, the results we obtain here can be generalised to the setting of general skew polynomial rings $R=K[t;\sigma,\delta]$, or to $R=D[t;\sigma]$, where $D$ is any finite-dimensional division algebra over its center. We refrained from doing so to make the already rather technical paper more accessible.

\emph{Acknowledgements:} The author thanks John Sheekey for providing her with a Magma code that computes  similar skew polynomials for small $q$, $n$ and $m$.

\section{Data availablilty statement}
No data were generated or used for the research described in this article.


\newcommand{\etalchar}[1]{$^{#1}$}
\providecommand{\bysame}{\leavevmode\hbox to3em{\hrulefill}\thinspace}
\providecommand{\MR}{\relax\ifhmode\unskip\space\fi MR }
\providecommand{\MRhref}[2]{%
  \href{http://www.ams.org/mathscinet-getitem?mr=#1}{#2}
}
\providecommand{\href}[2]{#2}


\begin{thebibliography}{DBU{\etalchar{+}}21}



\bibitem[Alb2018]{A} A. A. Albert, ``Modern higher algebra''. Courier Dover Publications, 2018.


\bibitem[AugLR2018]{ALR} D. Augot, P. Loidreau, G. Robert, \emph{Generalized Gabidulin codes over fields of any characteristic.} Designs Codes and Cryptography 86 (8) (2018), 1807-1848.
     \verb#https://doi.org/10.1007/s10623-017-0425-6#


 \bibitem[Brown2018]{BPhD} C. Brown 
  \emph{Petit algebras and their automorphisms}, PhD Thesis, University of Nottingham, 2018. 	\verb#https://eprints.nottingham.ac.uk/49613/#


     \bibitem[BrownPum2018]{BP}  C. Brown,  S. Pumpl\"un, \emph{The automorphisms of Petit's algebras}.
  Comm. Algebra  46 (2) (2018), 834-849.

\bibitem[BrowPum2018]{brown2018nonassociative} C. Brown, S. Pumpl\"un,
 \emph{How a nonassociative algebra reflects the properties of a skew polynomial.}
  Glasgow Mathematical Journal (Nov. 2019).
\\   \verb#https://doi.org/10.1017/S0017089519000478#


\bibitem[Car2017]{CaB} X. Caruso, J. Le Borgne \emph{A new faster algorithm for factoring skew polynomials over finite fields.} Journal of Symbolic Computation 79 (2), March–April 2017, 411-443.

\bibitem[Gab1985]{Gabidulin} E. M. Gabidulin, \emph{Theory of codes with maximum rank distance.} Problems of Information Transmission 21 (1985), 1–12.


\bibitem[GomLN2019]{GLN18} J. G\`{o}mez-Torrecillas,  F. J. Lobillo,; G. Navarro, \emph{Computing the bound of an Ore polynomial.
 Applications to factorization.} J. Symbolic Comput. 92 (2019), 269-297. 


\bibitem[Jac1996]{J96} N.~Jacobson, ``Finite-dimensional division algebras over fields.'' Springer
Verlag, Berlin-Heidelberg-New York, 1996.


\bibitem[LavrSh2013]{LS} M.~Lavrauw, J.~Sheekey, \emph{Semifields from skew-polynomial rings}. Adv. Geom. 13 (4) (2013), 583-604.

\bibitem[LunTrom2018]{LTZ} G. Lunardon, R. Trombetti, Y. Zhou, \emph{Generalized twisted Gabidulin codes}. J. Combinatorial Theory, Series A, 159 (2018), 79-106.


\bibitem[LobSanSh2026]{NewS}  F.J. Lobillo, P. Santonastaso, J.~Sheekey, \emph{Quotients of skew polynomial rings: new constructions of
division algebras and MRD codes}. J. Algebra 691 (2026), 648-693.

\bibitem[Ore1933]{O1} O. Ore, \emph{Theory of noncommutative polynomials.} Annals of Math. 34 (3) (1933), 480-508.


\bibitem[Pet1966]{P66} J.-C.~Petit, \emph{Sur certains quasi-corps g\'{e}n\'{e}ralisant un type d'anneau-quotient}.
 S\'{e}minaire Dubriel. Alg\`{e}bre et th\'{e}orie des nombres 20 (1966-67), 1-18.

 \bibitem[Pum2025]{Pum2025} S. Pumpl\"un,  \emph{The isotopy classes of Petit division algebras.}
 Preprint 2026
\\ \verb#arXiv:2511.18451# [math.RA]


\bibitem[Pum2025.2]{Pum2025.2} S. Pumpl\"un, \emph{A classification of the division algebras that are isotopic to a cyclic Galois field extension}.  To appear in Israel J. Math.
  \\  \verb#https://arxiv.org/abs/2407.11598#

 \bibitem[PumThom2022]{PT} S. Pumpl\"un, D. Thompson, \emph{The norm of a skew polynomial.}
 J. Algebra and Representation Theory 25 (2022), 869–887.
\\
 \verb#https://doi.org/10.1007/s10468-021-10051-z#

 \bibitem[PumThom2023]{PT2} S. Pumpl\"un, D. Thompson,  \emph{Division algebras and MRD codes from skew polynomials.}
   Glasgow Math. Journal 65 (2) (2023), 480--500.
  \verb#https://doi.org/10.1017/S001708952300006X#

   \bibitem[Pum2015]{P15} S. Pumpl\"un, {Albert's twisted field construction using division algebras with a multiplicative norm.}
J. Algebra Appl. 24 (2025), no. 11, Paper No. 2550253, 17 pp.
\\
\verb#https://doi.org/10.1142/S0219498825502536#


   \bibitem[Pum2017]{P18.0} S. Pumpl\"un, \emph{Finite nonassociative algebras obtained from skew polynomials  and possible applications to $(f,\sigma,\delta)$-codes}. Advances in Mathematics of Communications (AMC) 11 (3) (2017), 615-634.
    \verb#https://doi.org/10.3934/amc.2017046#


 \bibitem[Pum2018]{P18.1} S. Pumpl\"un, \emph{How to obtain lattices from $(f,\sigma,\delta)$-codes via a generalization of Construction A.}
Applicable Algebra in Engineering, Communication and Computing 29 (4) (2018), 313-333.
\verb#https://doi.org/10.1007/s00200-017-0344-9#


\bibitem[Sch1995]{Sch} R. D.~Schafer, ``An Introduction to Nonassociative Algebras.'' Dover Publ., Inc., New York, 1995.

\bibitem[Scha1885]{S} W. Scharlau, ``Quadratic and Hermitian Forms.'' Springer-Verlag, Berlin Heidelberg New York
Tokyo, 1985.


\bibitem[Sheek2019]{Sheekey}  J.~Sheekey
\emph{New semifields and new MRD codes from skew polynomial rings}, September 2019, Journal of the LMS
DOI: 10.1112/jlms.12281

\bibitem[Sheek2016]{Sheekey16} J. Sheekey, \emph{A new family of linear maximum rank distance codes.} Advances in Mathematics of Communications (AMC) 10 (2016), 475-488.
     \verb#https://doi.org/10.3934/amc.2016019#

\bibitem[Thom2021]{DT2020} D. Thompson, \emph{New classes of nonassociative division algebras and MRD codes.} PhD Thesis, University of Nottingham, 2021.
\verb#http://eprints.nottingham.ac.uk/64396/#


\end{thebibliography}
\end{document}